\documentclass[11pt,reqno]{article}
\usepackage{amsmath,amssymb,amsthm}
\usepackage{esint}
\usepackage{bm}
\usepackage{url}
\usepackage{subfigure}
\usepackage{graphicx,color}
\usepackage{algorithm}
\usepackage{algpseudocode}
\usepackage{algorithmicx}
\usepackage{hyperref}
\usepackage{float}
\usepackage{booktabs, multirow}
\usepackage{hhline}
\usepackage{setspace}

\usepackage{caption}
\usepackage{hyphsubst}
\usepackage{caption}

\usepackage{float}
\usepackage{filecontents}
\usepackage{hyperref}

\hypersetup{colorlinks,citecolor=blue,linkcolor=blue, urlcolor=blue}
\usepackage{pdflscape}
\usepackage{breakurl}

\newcommand{\RR}{\mathbb{R}}

\DeclareMathAlphabet{\itbf}{OML}{cmm}{b}{it}
\def\bp{{{\bf p}}}
\def\by{{{\itbf y}}}
\def\bx{{{\itbf x}}}
\def\hbx{{\hat{\bx}}}
\def\hby{{\hat{\by}}}

\def\be{{{\itbf e}}}

\def\bu{{{\itbf u}}}
\def\bv{{{\itbf v}}}
\def\bw{{{\itbf w}}}
\def\bd{{{\bf d}}}

\newcommand{\CC}{\mathbb{C}}
\newcommand{\OL}{\mathcal{L}}
\newcommand{\K}{\kappa}
\newcommand{\ds}{\displaystyle}
\newcommand{\bT}{\mathbf{T}}
\newcommand{\bnu}{\bm{\nu}}
\newcommand{\bGam}{\mathbf{\Gamma}}
\newcommand{\bI}{\mathbf{I}}
\def\bphi{{\boldsymbol{\varphi}}}
\newcommand{\cK}{\mathcal{K}}
\def\bpsi{{\boldsymbol{\psi}}}
\newcommand{\fC}{\mathfrak{C}}
\newcommand{\NN}{\mathbb{N}}
\newcommand{\bA}{\mathbf{A}}
\newcommand{\bB}{\mathbf{B}}
\newcommand{\bC}{\mathbf{C}}
\newcommand{\bL}{\mathbf{L}}
\newcommand{\bM}{\mathbf{M}}
\newcommand{\bN}{\mathbf{N}}
\newcommand{\bH}{\mathbf{H}}
\newcommand{\bJ}{\mathbf{J}}
\newcommand{\bq}{\mathbf{q}}
\newcommand{\bP}{\mathbf{P}}

\newcommand{\bQ}{\mathbf{Q}}
\newcommand{\bR}{\mathbf{R}}
\newcommand{\bS}{\mathbf{S}}
\renewcommand{\L}{\mathfrak{L}}

\newtheorem{theorem}{Theorem}[section]
\newtheorem{lemma}[theorem]{Lemma}
\newtheorem{proposition}[theorem]{Proposition}

\newtheorem{definition}[theorem]{Definition}
\newtheorem{remark}[theorem]{Remark}

\newcommand{\email}[1]{\protect\href{mailto:#1}{#1}}

\newcommand{\pathfigures}{Figures/}
\graphicspath{{\pathfigures}}

\begin{document}

\title{Three-dimensional Elastic Scattering Coefficients and Enhancement of the Elastic Near Cloaking
}

\author{
Hongyu Liu\footnotemark[1]\, \footnotemark[2]
\and
Wing Yan Tsui\footnotemark[3]
\and
Abdul Wahab\footnotemark[4]
\and
Xianchao Wang\footnotemark[5]
}
\maketitle
\renewcommand{\thefootnote}{\fnsymbol{footnote}}
\footnotetext[1]{Corresponding Author. E-mail address:  \email{hongyu.liuip@gmail.com}.}
\footnotetext[2]{Department of Mathematics, City University of Hong Kong, Kowloon, Hong Kong SAR, China (\email{hongyu.liuip@gmail.com}; \email{hongyliu@cityu.edu.hk}).}
\footnotetext[3]{Department of Mathematics, Hong Kong Baptist University, Kowloon, Hong Kong SAR, China (\email{wytsui.yan@gmail.com}).}
\footnotetext[4]{School of Sciences and Humanities, Nazarbayev University, 53 Kabanbay Batyr Ave. Nur-Sultan 010000, Kazakhstan (\email{abdul.wahab@nu.edu.kz}).}
\footnotetext[5]{Department of Mathematics, City University of Hong Kong, Kowloon, Hong Kong SAR, China, and School of Astronautics, Harbin Institute of Technology, Harbin 150001, Heilongjiang, China  (\email{xcwang90@gmail.com}).}
\renewcommand{\thefootnote}{\arabic{footnote}}

\begin{abstract}

This paper is concerned with the elastic near cloaking for the Lam\'e system in three-dimensions using the notion of elastic scattering coefficients (ESC). Accordingly, the ESC of arbitrary three-dimensional  objects are designed and some of their properties are discussed using elements of the elastic layer potential  theory. Then, near-cloaking structures, coined as ESC-vanishing-structures, are constructed for the elastic cloaking  at a fixed frequency or a band of frequencies. These multi-layered structures are designed so that their first few significant ESC vanish ahead of  transformation-elastodynamics. The invisibility effect is achieved as the arbitrary elastic object inside the cloak has near-zero scattering cross-section for a band of frequencies. The cloaking effect for the Lam\'e system is significantly enhanced by the proposed near-cloaking structures.

\end{abstract}



\section{Introduction}
Invisibility is always the subject that draws human curiosity. A region is said to be cloaked if its contents together with the cloak are invisible and undetectable in the background for a particular class of wave measurements. Many studies on cloaking appeared in recent past thanks to the convenience of the transformation-optics. Proposals of cloaking for conductivity equations with zero frequencies were given in 2003 by Greenleaf, Lassas, and Uhlmann (see \cite{greenleaf2003nonuniqueness, greenleaf2003anisotropic}). Other pioneering works on transformation-optics by Pendry, Schuring, and Smith \cite{pendry2006controlling} and by Leonhardt \cite{leonhardt2006optical}  in 2006 gave a design for singular transformation to make objects invisible to electromagnetic waves.

Recently, the idea of transformation-optics was also applied in elastic cloaking theory \cite{hu2015nearly, diatta2014controlling, brun2009achieving, norris2011elastic, milton2006cloaking, parnell2012nonlinear, diatta2013cloaking}.  However, it is more difficult to achieve invisibility effect in elasticity than the Maxwell equations or the Helmholtz equation for the optical waves since the Lam\'e system governing elastic wave propagation lacks invariance \cite{milton2006cloaking}. Towards this end, the idea of transformation-elastodynamics was proposed by Hu and Liu \cite{hu2015nearly} in 2015 which is similar to the transformation-optics. A diffeomorphism is designed that blows up a point to form a \emph{hole} (the cloaked region) and compresses the ambient space around the point to form a \emph{shell} (the cloaking layer).
Using the transformation-elastodynamics, an illusion is created by the customized effects on wave propagation. As the scattering measurements between the virtual space (the homogeneous background with the point) and the physical space (composed of the cloaking layer, the cloaked region, and the interface between them) are coincident, everything can be hidden to achieve an invisibility effect. Thus, the cloaking layer and anything inside it becomes hidden from the observers outside the cloaking layer.

The blow-up-a-point scheme considers a singular transformation mapping which introduces singularities for the elastic material tensors. Such singularities make both the mathematical analysis and physical application difficult. In order to avoid the singular structures, a regularized approximate cloaking scheme is considered. The general framework for constructing elastic near-cloak using the transformation-elastodynamics is introduced by Hu and Liu in \cite{hu2015nearly}.

In 2013,  Ammari et al. \cite{ammari2013enhancement} showed that near-cloaking from boundary measurements for the conductivity equation can be enhanced by canceling significant generalized polarization tensors (GPTs) of the cloaking device. It is established that the GPTs-vanishing structures (multi-layered coating with vanishing GPTs) can be combined with the transformation-optics, making the inclusion nearly invisible from the far-field measurements.  Ammari and his coworkers also showed that similar enhancement schemes for near-cloaking can also be applied in the Helmholtz equation and Maxwell equations. For doing so, they considered the scattering cross-sections and expressed the scattering amplitudes in terms of the so-called \emph{scattering coefficients}. Their multi-layered structures were designed to cancel out the scattering coefficients up to an arbitrary order. Applying the transformation-optics to the S-vanishing structures, it is established that one can enhance the invisibility effect. The numerical experiments also confirmed their results \cite{ammari2012enhancement}. Recently, Abbas et al. \cite{abbas2016elastic} discussed ESC in two-dimensions and designed ESC-vanishing structures for the enhancement of the near-elastic cloaking frameworks. The study was focused on the design of two-dimensional ESC, their properties, and their reconstruction from multi-static response measurements. They briefly discussed the use of ESC for enhancing the performance of near-cloaking devices.

In this work, our aim is to apply the idea of regularized approximate elastic cloaks in \cite{hu2015nearly}  to improve the near-cloaking result. We extend the method of \cite{abbas2016elastic, ammari2012enhancement, ammari2013enhancement, ammari2013enhancementHelmoltz, ammari2013enhancementMaxwell} to the three-dimensional elastic scattering problem in order to achieve drastically enhanced invisibility effect from the scattering cross-section measurements at a fixed frequency via transformation-elastodynamics approach. We first design three-dimensional ESC of an arbitrary elastic inclusion and introduce new multi-layered structures around the small ball in the virtual space. After vanishing the first few terms of scattering coefficients, we apply the transformation-elastodynamics to achieve invisibility. We show that arbitrary elastic object inside the cloaked region has near-zero scattering cross-section for a band of frequencies. We also substantiate that the new near-cloaking structures can enhance the cloaking effect for the Lam\'e  system significantly. The three-dimensional ESC are not the straightforward extension of two-dimensional ESC and the design of ESC-vanishing structures in three-dimensions is much more involved than the two-dimensional case. Moreover, this study is focused on near-cloaking.

The rest of the paper is organized as follows. In Section \ref{sect:LPT}, we recall a few fundamental results from the layer potential theory for scattering in linear isotropic elastic media. In Section \ref{sect:ESC}, we derive the multipole expansion of the solution to the Lam\'e  system and introduce three-dimensional  ESC of an arbitrary elastic inclusion. Some of the properties of ESC are also discussed. The multi-layered structures with vanishing ESC are designed in Section \ref{sect:SVS}. Section \ref{sect:Cloaking} deals with nearly ESC-vanishing structures and the procedure for the enhancement of nearly elastic cloaking using transformation-elastodynamics.
Finally, the contributions of the paper are summarized in Section \ref{sect:Conc}.

\section{Elements of Layer Potential Theory}\label{sect:LPT}

The primary concern of this article is the enhancement of the near-elastic cloaking devices using the concept of the ESC. Accordingly, our methodology is based on the integral formulations of the scattered elastic fields. Therefore, it is important to introduce the key components of the integral formulation in linear isotropic time-harmonic elasticity and provide some background material first. For details beyond those provided in this section,  the readers are suggested to consult the monographs \cite{Princeton, Bergh, Kup1, Kup2}.

\subsection{Preliminaries and Notation}

Let $D\subset\RR^3$ be a smooth, open, and bounded domain with a connected Lipschitz boundary $\partial D$.   Let $L^2(D)$ be the space of square-integrable functions defined in the usual way with norm
\begin{align*}
\|v\|_{L^2(D)}:=\left(\int_{D}|v|^2d\bx\right)^{1/2}.
\end{align*}
We define the Hilbert space $H^1(D)$ by
\begin{align*}
H^1(D)=\left\{v \in L^2 (D)|\quad \nabla v \in L^2 (D)\right\},
\end{align*}
equipped with the norm
\begin{align*}
\|v\|_{H^1(D)}=\left (\int_D |v|^2 dx+\int_D|\nabla v|^2 d\bx\right)^{1/2}.
\end{align*}
We also define the Hilbert space $H^2(D)$  containing the functions $v\in H^1(D)$ such that $\partial_{pq} v\in L^2(D)$ for all $p,q\in\{1,2,3\}$. We denote the interpolation space $[H^1(D), H^2(D)]_{1/2}$ by $H^{3/2}(D)$. Let $\{\mathbf{t}_1,\mathbf{t}_2\}$ form an orthonormal basis for the tangent plane to $\partial D$ at point $\bx$ and $\partial/\partial \mathbf{t}:=\sum_{i=1}^2(\partial/\partial \mathbf{t}_i )\mathbf{t}_i$ be the tangential derivative at $\bx\in\partial D$. Then, we say that $v\in H^1(\partial D)$ if $v\in L^2(\partial D)$ and $\partial v/\partial \mathbf{t}\in L^2(\partial D)$. We refer the interested readers to \cite{Bergh} for details.

Let $D$ be occupied by a homogeneous isotropic linear elastic material having compression modulus $\lambda_{1}$, shear modulus $\mu_{1}$, and volume density $\rho_1\in\RR_+$. Let the background domain $\RR^3\setminus\overline{D}$ be occupied by a different homogeneous isotropic linear elastic material having corresponding Lam\'e parameters $(\lambda_{0},\mu_{0})$ and volume density $\rho_0\in\RR_+$.  It is supposed that the interior and exterior Lam\'e parameters satisfy the strong convexity conditions, i.e.,
\begin{align}
\mu_{0}>0,
\quad 3\lambda_{0}+2\mu_{0}>0,
\quad \mu_{1}>0,
\quad \text{and}\quad
 3\lambda_{1}+2\mu_{1}>0.
\label{eq:convex}
\end{align}
Moreover, it is assumed that
\begin{align}
\label{Lame-Conditions}
 (\mu_{0} -\mu_{1})(\lambda_{0} -\lambda_{1} )\geq 0
 \quad \text{and} \quad
 (\mu_{0} -\mu_{1})^2+(\lambda_{0} -\lambda_{1} )^2\neq 0.
\end{align}
Let us define the material parameters of the background medium in the presence of $D$ by
\begin{align}
[\lambda; \mu; \rho](\bx)=[\lambda_0;\mu_0;\rho_0]\chi_{\RR^3\setminus\overline{D}}(\bx)+[\lambda_1;\mu_1;\rho_1]\chi_{D}(\bx),
\end{align}
where $\chi_D$ represents the characteristic function of the domain $D$.

Let $\omega\in\RR_+$ denote the frequency of the mechanical oscillations and let $\bu^{\rm inc}$ be an incident time-harmonic elastic field impinging on $D$ with time-variations $e^{-i\omega t}$ being suppressed, i.e.,
\begin{equation}
\label{eq:Elasto}
\nabla\cdot{(\fC^0:\nabla\bu^{\rm inc})}+\omega^{2}  \rho_0 \bu^{\rm inc}=\mathbf{0}, \quad \text{in }\RR^3.
\end{equation}
Here, $\fC^0=(C^0_{ijkl})_{i,j,k,l=1}^3$ is the background elasticity tensor defined by
\begin{align*}
C^0_{ijkl}=\mu_0(\delta_{ik}\delta_{jl}+\delta_{il}\delta_{jk})+\lambda_0\delta_{ij}\delta_{kl},
\end{align*}
where $\delta_{ij}$ is the Kronecker's delta function and `$:$' is the contraction operator defined by
\begin{align*}
\fC^0:\mathbf{A}:=\sum_{k,l=1}^3 C^0_{ijkl} a_{kl},
\end{align*}
for arbitrary matrix $\mathbf{A}=(a_{ij})_{i,j=1}^3$.  For ease of notation, we define the linear isotropic elasticity operator corresponding to the Lam\'e parameters $(\lambda_0,\mu_0 )$ by
\begin{align*}
\OL_{\lambda_0,\mu_0}[\bw]:=(\fC^0:\nabla\bw)=\mu_0 \nabla \bw+(\lambda_0+\mu_0) \nabla\nabla\cdot \bw,
\end{align*}
for an arbitrary smooth function $\bw:\RR^3\to\RR^3$.  Without loss of generality, we assume that $\rho_0=1$, henceforth.

Let $\bu^{\rm sc}$ be the scattered field generated by the interaction of the incident field $\bu^{\rm inc}$ with $D$. We define the pressure and shear parts of  $\bu^{\rm sc}$, respectively, by
\begin{align*}
\bu^{\rm sc}_{P}(\bx) :=\ds  -\frac{1}{\K_P^2}\nabla\nabla\cdot\bu^{\rm sc}(\bx)
\quad\text{and}\quad
\bu^{\rm sc}_{S}(\bx) :=\frac{1}{\K_S^2}
 \ds \nabla\times\nabla\times\bu^{\rm sc}(\bx),
\end{align*}
for all $\bx\in \RR^3\setminus{\overline{D}}$. Here, the constants $\K_P$ and $\K_S$   are the longitudinal and transverse wave-numbers, respectively, i.e.,
\begin{align*}
\K_\alpha:= \frac{\omega}{c_\alpha},
\quad\text{where}\quad c_P= \sqrt{{\lambda_0+2\mu_0}},\quad c_S= \sqrt{{\mu_0}}, \quad\text{and}\quad \alpha=P,S.
\end{align*}
It can be easily verified that  $\bu^{\rm sc}_{P}$ and $\bu^{\rm sc}_{S}$ satisfy
\begin{align*}
&\left(\Delta+\K_P^2\right)\bu^{\rm sc}_{P}=\mathbf{0}
\quad\text{and}\quad
\left(\Delta+\K_S^2\right)\bu^{\rm sc}_{S}=\mathbf{0},  \quad\text{in }\RR^3\setminus\overline{D},
\\
&
\text{ such that }\nabla\cdot \bu^{\rm sc}_{S}=0
\quad\text{and}\quad
\nabla\times \bu^{\rm sc}_{P}=\mathbf{0}.
\end{align*}
The scattered field  is said to satisfy the \emph{Kupradze radiation condition} if
\begin{align*}
\ds\lim_{|\bx|\to+\infty} |\bx|\left(\frac{\partial \bu^{\rm sc}_{\alpha}}{\partial |\bx|}-\iota\K_\alpha\bu^{\rm sc}_{\alpha}\right)=\mathbf{0}, \quad \alpha=P,S, \quad \iota=\sqrt{-1},
\end{align*}
uniformly in all directions $\hbx\in\mathbb{S}^{2}:=\{\bx\in\RR^3\,:\,\, |\bx|=1\}$. Here, $\hbx:=\bx/|\bx|$ for any $\bx\in\RR^3\setminus\{\mathbf{0}\}$ and $\partial/\partial |\bx|$ denotes the derivative in the radial direction.

Let $\bu^{\rm tot}:=\bu^{\rm sc}+\bu^{\rm inc}$ be the total field generated by the interaction of the incident field $\bu^{\rm inc}$ with $D$. Then,  $\bu^{\rm tot}$ is the solution to the Lam\'e system
\begin{align}
\label{elastodynamics equations}
\begin{cases}
\nabla\cdot(\fC:\nabla\bu^{\rm tot}(\bx))+\omega^{2}  \rho \bu^{\rm tot}(\bx) =0,  &\bx \text{ in } \RR^3,
\\
\bu^{\rm sc}(\bx) \quad\text{satisfies the Kupradze radiation condition as} & |\bx|\to+\infty.
\end{cases}
\end{align}
Here, $\fC$ is the elasticity tensor of the elastic formation in the presence of $D$ and is defined in terms of the piece-wise constant parameters $(\lambda,\mu,\rho)$.  It is well known that the scattering problem \eqref{elastodynamics equations}  is well-posed in $H^1(\RR^3\setminus\partial D)^3$ (see, e.g., \cite{Kup1, Kup2}).

Henceforth, the surface traction, denoted by $\bT[\bw]$, associated to the Lam\'e parameters $(\lambda_0,\mu_0)$ is defined by
\begin{equation}
 \label{traction}
\bT[\bw]:=\bnu\cdot(\fC:\nabla \bw)=\lambda_0(\nabla\cdot\bw)\bnu+\mu_0(\nabla \bw+\nabla\bu^\top )\bnu,
\end{equation}
for an arbitrary vector field $\bw:\RR^3\to\CC^3$. Here, \(\nabla \bw\) denotes the Jacobian matrix of $\bw$, $\bnu$ is the outward unit normal vector at the boundary $\partial D$, and the superposed $\top$ indicates the matrix transpose. At the interface $\partial D$, the solution $\bu^{\rm tot}$ satisfies the transmission conditions
\begin{equation}
\label{transmission condition}
\bu^{\rm tot} |_+={\bu^{\rm tot}|_-}
\quad\text{and}\quad
{\bT[\bu^{\rm tot}]\Big|_+}= {\bT[\bu^{\rm tot}]\Big|_-},
\end{equation}
where
$$
v\big|_\pm(\bx)=\lim_{\epsilon\to 0^+} v(\bx\pm \epsilon\bnu), \quad \bx\in\partial D,
$$
for an arbitrary function $v$.

\subsection{Integral Representation of Scattered Field}

In this section, we summarize the layer potential technique for time-harmonic linear elasticity in order to facilitate the definition of the scattering coefficients of $D$. Towards this end, we first introduce the Kupradze matrix $\bGam^\omega$ as the fundamental solution of the time-harmonic elasticity equation in $\RR^3$ with parameters $(\lambda_0,\mu_0,\rho_0)$, i.e., $\bGam^\omega$ satisfies
\begin{align*}
\nabla\cdot \left(\fC^0:\nabla \bGam^\omega(\bx)\right)+\omega^2\rho_0\bGam^\omega (\bx)= -\delta_{\mathbf{0}}(\bx)\bI_3, \qquad \bx\in\RR^3,
\end{align*}
subject to Kupradze's outgoing radiation conditions. Here, $\delta_{\by}$ is the Dirac mass at $\by\in\RR^3$ and $\bI_3\in\RR^{3\times 3}$ is the identity matrix. It is well known that (see, for instance, \cite{morse1953methods})
\begin{equation}
\label{Green_fun}
\bGam^\omega(\bx)=\frac{1}{\mu_0}\left[\ds\left(\bI_2+\frac{1}{\K_S^2}\nabla\nabla^\top\right)g(\bx,\K_S)-\frac{1}{\K_S^2}\nabla\nabla^\top g(\bx,\K_P)\right],\quad \bx\in\RR^2\setminus\{ \mathbf{0}\}.
\end{equation}
The function $g(\cdot,\K_\alpha)$, for $\alpha=P,S$, is the fundamental solution to the Helmholtz operator $-(\Delta+\K_\alpha^2)$  in $\RR^3$ with wave-number $\K_\alpha\in \RR_+$ ($\K_\alpha=\K_P$ or $\K_S$), i.e.,
$$
\Delta g(\bx,\K_\alpha)+\K_\alpha^2g(\bx,\K_\alpha)=-\delta_{\bf 0}(\bx),\quad \bx\in\RR^3,
$$
subject to \emph{Sommerfeld's} outgoing radiation condition
$$
\lim_{|\bx|\to+\infty}|\bx|\left[\frac{\partial g(\bx,\K_\alpha)}{\partial |\bx|}-\iota\K_\alpha g(\bx,\K_\alpha)\right]=0, \qquad \bx\in\RR^3.
$$
It is easy to see that
\begin{equation}
g(\bx,\K_\alpha)=\ds\frac{1}{4\pi |\bx|} e^{\iota\K_\alpha|\bx|},\quad\forall\bx\in\RR^3\setminus\{\mathbf{0}\}.
\label{green-fn-g}
\end{equation}
In the sequel, we denote  $\bGam^\omega(\bx,\by):=\bGam^\omega(\bx-\by)$.

Let us now introduce the elastic single-layer potential by
\begin{eqnarray*}
\mathcal{S}_D^\omega[\bphi](\bx) :=\ds\int_{\partial D}\bGam^\omega(\bx,\by)\bphi(\by)d\sigma(\by),\quad\forall \bx\in\RR^3\setminus\partial D,
\end{eqnarray*}
for all densities $\bphi\in L^2(\partial D )^3$.
Here and throughout this article, $d\sigma$ denotes the infinitesimal boundary differential element.   The  traces $ \mathcal{S}^\omega_D[\bphi]\big|_{\pm}$ and $\bT\left[\mathcal{S}_D^\omega[\bphi]\right]\big|_{\pm}$  across the interface $\partial D$ are well-defined and satisfy the jump conditions (see, for instance, \cite{Dahlberg})
\begin{eqnarray}
\label{Sjumps}
\begin{cases}
\ds\mathcal{S}_D^\omega[\bphi]\big|_{+}(\bx)=\mathcal{S}_D^\omega[\bphi]\big|_{-}(\bx),
\\
\ds \bT\left[\mathcal{S}_D^\omega[\bphi]\right]\big|_{\pm}(\bx)
=\left(\pm\frac{1}{2}\mathcal{I}+(\mathcal{K}_D^\omega)^*\right)[\bphi](\bx), \quad{\rm a.e.}\quad\bx\in\partial D,
\end{cases}
\end{eqnarray}
where the boundary integral operator $(\cK_D^\omega)^*$ is defined by
$$
(\cK_D^\omega)^*[\bphi](\bx) = {\rm p.v.}\ds\int_{\partial D}\bT\left[\bGam^\omega\right](\bx,\by)\bphi(\by)d\sigma(\by),\quad{\rm a.e.}\quad\bx\in\partial D,
$$
for all $\bphi\in L^2(\partial D)^3$ and  p.v.  stands for the Cauchy principle value of the integral. We precise that the surface traction of matrix  $\bGam^\omega$ is defined column-wise, i.e., for all constant vectors $\mathbf{p}\in\RR^3$,
$$
\bT\left[\bGam^\omega\right]\mathbf{p}=\bT\left[\bGam^\omega\mathbf{p}\right].
$$

In the layer potential technique, the total displacement field $\bu^{\rm tot}$ in the presence of inclusion $D$ is first represented in terms of the single-layer potentials  $\mathcal{S}^\omega_D$ and $\widetilde{\mathcal{S}}^\omega_D$ of unknown densities  $\bphi, \bpsi\in L^2(\partial D)^3$  as (see \cite[Theorem 1.8]{Princeton})
\begin{equation}\label{u-int-rep}
\bu^{\rm tot}(\bx,\omega)=
\begin{cases}
\bu^{\rm inc}(\bx,\omega)+\mathcal{S}^\omega_D[\bpsi](\bx,\omega), & \bx\in\RR^3\setminus\overline{D},
\\
\widetilde{\mathcal{S}}^\omega_D[\bphi](\bx,\omega), & \bx\in D.
\end{cases}
\end{equation}
Then, the densities  $\bphi, \bpsi\in L^2(\partial D)^3$ are uniquely sought by solving the system of integral equations
\begin{equation}
\label{integral-system}
\begin{pmatrix}
\widetilde{\mathcal{S}}_D^\omega & -{\mathcal{S}}_D^\omega
\\
\ds
\widetilde{{\bT}}[\widetilde{\mathcal{S}}_D^\omega]\Big|_-
&
\ds
-\bT[\mathcal{S}_D^\omega]\Big|_+
\end{pmatrix}
\begin{pmatrix}
\bphi
\\
\bpsi
\end{pmatrix}
=
\ds
\begin{pmatrix}
\bu^{\rm inc}
\\
\ds\bT[\bu^{\rm inc}]
\end{pmatrix}\Bigg|_{\partial D}.
\end{equation}
Here, the superposed $\sim$ is used to distinguish the  single-layer potential and the surface traction defined using the interior parameters $(\lambda_1,\mu_1,\rho_1)$. To simplify the matters, the dependence of $\bu^{\rm inc}$, $\bu^{\rm sc}$, $\bu^{\rm tot}$, $\bphi$, and $\bpsi$ on frequency $\omega$ is suppressed unless it is necessary.

The following result from \cite[Theorem 1.7]{Princeton} guarantees the unique solvability of system \eqref{integral-system} and consequently that of problems \eqref{elastodynamics equations} and \eqref{u-int-rep}.

\begin{theorem}
\label{thmSolve}
Let $D$ be a Lipschitz bounded domain in $\RR^3$ with parameters $0<\lambda_1,\mu_1,\rho_1<\infty$ satisfying condition \eqref{Lame-Conditions} and let $\omega^2\rho_1$ be different from Dirichlet eigenvalues of the operator $-\OL_{\lambda_1,\mu_1}$ on $D$. For any function $\bu^{\rm inc}\in H^1(\partial D)^3$, there exists a unique solution $(\bphi,\bpsi)\in L^2(\partial D)^3\times L^2(\partial D)^3$ to the integral system \eqref{integral-system}.
Moreover, there exists a constant $C>0$ such that
\begin{align}
\label{stability}
\|\bphi\|_{L^2(\partial D)^3}+\|\bpsi\|_{L^2(\partial D)^3}
\leq C \left(\|\bu^{\rm inc}\|_{H^1(\partial D)^3}+\left\|\bT[\bu^{\rm inc}]\right\|_{L^2(\partial D)^3}\right).
\end{align}
\end{theorem}

\section{Three-dimensional ESC}\label{sect:ESC}

In this section, we seek the multipolar expansions of the scattered elastic field and the Kupradze matrix in terms of the spherical elastic waves. Our main goal is to introduce the ESC of an inclusion and highlight some of their important features.

\subsection{Multipolar Expansions of Elastic Fields}

Let $\bx=(x_1,x_2,x_3)$ be a point in the Cartesian coordinate system and be equivalently defined by $(r\sin\theta_\bx\cos\varphi_\bx, r\sin\theta_\bx\sin\varphi_\bx, r\cos\theta_\bx)$  in spherical coordinate system with $r:=|\bx|$, $\theta_\bx\in[0,\pi]$, and $\varphi_\bx\in[0,2\pi)$.  Let $\{\be_r, \be_\theta, \be_\varphi\}$ be the standard orthonormal basis for the spherical coordinate system, where
\begin{align}
\be_r:=\frac{\bx}{r},
\quad
\be_\theta:=\left(\cos\theta\cos\varphi, \cos\theta\sin\varphi, -\sin\theta\right),
\quad
\be_\varphi:=\left(-\sin\varphi, \cos\varphi, 0\right).
\end{align}

For all integers $n\geq 0$ and $m=-n,\cdots, n$,  let $Y_{nm}$ denote the $nm$-th spherical harmonics defined on the unit sphere $\mathbb{S}^2$ by
\begin{eqnarray*}
Y_{nm}(\hbx):=Y_{nm}(\theta_\bx,\varphi_\bx)=(-1)^{m}\sqrt{\frac{(n-m)!(2n+1)}{4\pi(n+m)!}}P_n^{m}(\cos \theta_\bx)e^{im\varphi_\bx},
\end{eqnarray*}
where $P^{m}_n$ is the $m$-th associated Legendre function of order $n$  (see, e.g., \cite[2.4.78]{nedelec}). The vector spherical harmonics are then defined as
\begin{eqnarray*}
\begin{cases}
\ds\bA_{nm}(\hbx):=\be_{r} Y_{nm}(\theta_\bx,\varphi_\bx), & n\geq 0,
\\
\ds\bB_{nm}(\hbx):=\frac{1}{\sqrt{n(n+1)}}\nabla_{\mathbb{S}^2} Y_{nm}(\theta_\bx,\varphi_\bx), & n>0,
\\
\ds\bC_{nm}(\hbx):= \bB_{nm}(\theta_\bx,\varphi_\bx)\times \be_r, & n>0,
\end{cases}
\end{eqnarray*}
where $\nabla_{\mathbb{S}^2}$ denotes the surface gradient on $\mathbb{S}^2$, i.e.,
\begin{align*}
\nabla_{\mathbb{S}^2}=\be_\theta\frac{\partial}{\partial\theta} +\be_\varphi\frac{1}{sin\theta}  \frac{\partial}{\partial \varphi}.
\end{align*}
It is well known that the vector spherical harmonics $\{\bA_{nm},\bB_{nm},\bC_{nm}\}$ form a complete orthonormal basis of  $L^2(\mathbb{S}^2)^3$ (see, e.g., \cite[p. 1900]{morse1953methods}). In fact, we have the orthogonality conditions,
\begin{align}
&\int_{\Omega}\bA_{nm}\cdot \overline{\bB_{kl}}d\Omega=\int_{\Omega}\bA_{nm}\cdot \overline{\bC_{kl}}d\Omega=\int_{\Omega}\bB_{nm}\cdot \overline{\bC_{kl}}d\Omega=0,\quad n,k>0,
\label{othogonal condition 1}
\\
&\int_{\Omega}\bA_{nm}\cdot \overline{\bA_{kl}}d\Omega=\int_{\Omega}\bB_{nm}\cdot \overline{\bB_{kl}}d\Omega=\int_{\Omega}\bC_{nm}\cdot \overline{\bC_{kl}}d\Omega= \delta_{nk}\delta_{ml},
\label{othogonal condition 2}
\end{align}
for all $n\geq 0 $, $m=-n,\cdots, n$, $k\geq 0$, and $l=-k,\cdots, k$.  Here, the integration is over a spherical surface $\Omega$ of unit radius with  differential element $d\Omega = \sin d\theta d\varphi$ and angles $(\theta,\varphi)\in[0,\pi]\times[0,2\pi)$.

Let  $j_n$ and $h^{(1)}_n$ be the order $n$ spherical Bessel and first kind Hankel functions, respectively. For each $\K_\alpha\in\RR_+$, $\alpha=P,S$, $n\geq 0$, and $m=-n,\cdots,n$, we  construct the functions $v_{nm}(\cdot,\K_\alpha)$ and $w_{nm}(\cdot,\K_\alpha)$ by
\begin{align*}
\begin{cases}
v_{nm}(\bx,\K_\alpha):=h^{(1)}_n(\K_\alpha r)Y_{nm}(\theta_\bx,\varphi_\bx),
\\
w_{nm}(\bx,\K_\alpha):=j_n(\K_\alpha r)Y_{nm}(\theta_\bx,\varphi_\bx).
\end{cases}
\end{align*}
The functions $v_{nm}(\cdot, \K_\alpha)$ and $w_{nm}(\cdot, \K_\alpha)$ are called the \emph{Debye's potentials} and are the outgoing radiating and entire solutions to the Helmholtz equation $\Delta v+\K_\alpha^2 v=0$ in $\RR^3\setminus\{\mathbf{0}\}$ and $\RR^3$, respectively (see, e.g., \cite[Theorem 9.14]{Monk}).

Using vector spherical harmonics  $\bA_{nm}$, $\bB_{nm}$, $\bC_{nm}$ and scalar Debye's potentials $v_{nm}$, $w_{nm}$, we define the exterior and interior vector Debye's potentials, for all $\K_P,\K_S\in\RR_+$, $n\in\NN\cup\{0\}$, and $m=-n,\cdots,n$,  respectively, as
\begin{align}
\label{interior multipole fields}
 \begin{cases}
 \ds {\bL_{nm}}=\ds\frac{1}{\K_P}\nabla[v_{nm}(\bx,\K_P)]
=(h_n^{(1) } )^\prime (\K_P r) \bA_{nm}+\frac{ h_n^{(1) } (\K_P r)}{\K_P r} \sqrt{n(n+1)} {\bB_{nm}},  & n\geq 0,
  \\
\ds{\bM_{nm}}=\nabla\times[\bx v_{nm}(\bx,\K_S)]
 =h_n^{(1)} (\K_S r) \sqrt{n(n+1)} \bC_{nm}, &n>0,
  \\
\ds{\bN_{nm}}=\frac{1}{\K_S}  \nabla\times{\bM_{nm}}
= \frac{n(n+1)}{\K_Sr} h_n^{(1) } (\K_S r) {\bA_{nm}}+\frac{\sqrt{n(n+1) }}{\K_S r} \mathcal{H}_n (\K_S r) {\bB_{nm}},  &n>0 ,
\end{cases}
\end{align}
and
\begin{align}
\label{exterior multipole fields}
 \begin{cases}
 \ds {\widetilde{\bL}_{nm}}=\ds\frac{1}{\K_P}\nabla[w_{nm}(\bx,\K_P)]
=(j_n)^\prime (\K_P r) \bA_{nm}+\frac{ j_n (\K_P r)}{\K_P r} \sqrt{n(n+1)} {\bB_{nm}}, &n\geq 0,
  \\
\ds{\widetilde{\bM}_{nm}}=\nabla\times[\bx w_{nm}(\bx,\K_S)]
 =j_n (\K_S r) \sqrt{n(n+1)} \bC_{nm}, &n>0,
  \\
\ds{\widetilde{\bN}_{nm}}=\frac{1}{\K_S} \nabla\times{\widetilde{\bM}_{nm}}
= \frac{n(n+1)}{\K_Sr}j_n (\K_S r) {\bA_{nm}}+\frac{\sqrt{n(n+1) }}{\K_S r} \mathcal{J}_n (\K_S r) {\bB_{nm}},   &n>0,
\end{cases}
\end{align}
where $\mathcal{H}_n (t)=h_n^{(1) } (t)+t{(h_n^{(1)}) }^\prime(t)$ and $\mathcal{J}_n(t)=j_n(t)+t(j_n)^\prime(t)$ (see, e.g., \cite[p.1865-66]{morse1953methods}). Here and throughout in this paper, a prime over a function denotes the derivative with respect to its argument. It is easy to verify that the functions $\bL_{nm}$, $\bM_{nm}$, and $\bN_{nm}$  are the radiating solutions to the Lam\'e equation in $\RR^3\setminus\{\mathbf{0}\}$ and  $\widetilde{\bL}_{nm}$, $\widetilde{\bM}_{nm}$, and $\widetilde{\bN}_{nm}$  are the entire solutions to the Lam\'e equation in $\RR^3$.

Let $\bp$ be a fixed vector in $\RR^3$. Then, for all  $|\bx|>|\by|$ (see, e.g.,  \cite[Eqs. 6.73--6.74]{colton2003inverse}),
\begin{align}
\label{green function}
\bGam^\omega (\bx-\by)\bp = &\sum_{n=0}^\infty \frac{\iota\K_P}{c_P^2} \sum_{m=-n}^n\bL_{nm} (\bx,\K_P) \overline{\widetilde{\bL}_{nm}  (\by,\K_P)} \cdot \bp
\nonumber
\\
&+\sum_{n=1}^\infty \frac{\iota\K_S}{ n(n+1)c_S^2} \sum_{m=-n}^n\bM_{nm} (\bx,\K_S) \overline{\widetilde{\bM}_{nm}  (\by,\K_S)} \cdot \bp
\nonumber
\\
&+\sum_{n=1}^\infty \frac{\iota\K_S}{ n(n+1)c_S^2} \sum_{m=-n}^n\bN_{nm}  (\bx,\K_S) \overline{\widetilde{\bN}_{nm} (\by,\K_S)} \cdot \bp.
 \end{align}
Similarly, the plane elastic waves can be also expressed in terms of the vector Debye's potentials. In this article, we will consider the incident elastic waves of the form of either a plane shear wave
\begin{align*}
\bu^{\rm inc}(\bx)=\bu_S^{\rm inc}(\bx):=\bq e^{\iota\K_S \bx \cdot \bd},
\end{align*}
 or a plane pressure wave
\begin{align*}
\bu^{\rm inc}(\bx) = \bu_P^{\rm inc}(\bx):=\bd e^{\iota\K_P \bx\cdot \bd},
\end{align*}
where $\bd\in\mathbb{S}^2$ is the direction of incidence and $\bq\in\mathbb{S}^2$ is any vector orthogonal to $\bd$, i.e., $\bq\cdot \bd=0$.
Using the vector version of the Jacobi-Anger expansion, the multipolar expansion of the  plane shear wave is derived as  (see, e.g., \cite[Theorem 6.26]{colton2003inverse})
\begin{align}
\bu_S^{\rm inc}(\bx)
=
-\sum_{k=1}^\infty \frac{4\pi\iota^k}{\sqrt{k(k+1)}} \sum_{l=-k}^k
&\Bigg((\overline{\bC_{kl}(\bd)}\cdot \bp) \widetilde{\bM}_{kl}  (\bx,\K_S)
+\iota(\overline{\bB_{kl}(\bd)}\cdot \bp) \widetilde{\bN}_{kl} (\bx,\K_S)\Bigg),
\label{incident s-wave}
 \end{align}
 where $\bp=(\bd\times\bq)\times\bd$ with $\bq\in\RR^3$ being the polarization direction.  Similarly, the expansion of the plane pressure wave solution to the elastodynamics equation can be derived as (see, e.g., \cite{johnson1965numerical})
\begin{equation*}
\bu^{\rm inc}_P(\bx)=-\sum^\infty_{k=0}4\pi \iota^{k+1}\sum_{l=-k}^k\left(\left(\overline{\bA_{kl}(\bd)}\cdot\bd\right)\bL_{kl}(\bx,\K_P)\right).
 \end{equation*}

\subsection{Scattering Coefficients of Elastic Inclusions}

For ease of presentation, we shall make use of the notation
\begin{align*}
&\bH_{kl}^L(\bx):=\bL_{kl}(\bx,\K_P),\quad \bH^M_{kl}(\bx):=\bM_{kl}(\bx,\K_S), \quad \bH^N_{kl}(\bx):=\bN_{kl}(\bx,\K_S),
\\
&\bJ_{kl}^L(\bx):=\widetilde{\bL}_{kl}(\bx,\K_P),\quad \bJ^M_{kl}(\bx):=\widetilde{\bM}_{kl}(\bx,\K_S), \quad \bJ^N_{kl}(\bx):=\widetilde{\bN}_{kl}(\bx,\K_S).
\end{align*}
Moreover, we reserve the notation $\jmath, \jmath'= L, M, N$, and introduce $n_\jmath$ and $k_\jmath$ such that $n_\jmath, k_\jmath=0$ if $\jmath=L$ and $n_\jmath,k_\jmath=1$ otherwise.

The multipolar expansion \eqref{green function} of the Kupradze matrix $\bGam^\omega$ renders the series representation of the single-layer potential in terms of $\bJ^\jmath_{nm}$ and $\bH^{\jmath}_{nm}$, and subsequently, leads to the multipolar expansion of the scattered field $\bu^{\rm sc}$ therefrom thanks to the integral representation (\ref{u-int-rep}). Specifically, for $\bx\in\RR^3\setminus\overline{D}$ sufficiently far from the boundary $\partial D$,
\begin{align*}
\bu^{\rm sc}(\bx)
=&
\sum_{n=0}^\infty \frac{\iota\K_P}{c_P^2} \sum_{m=-n}^n \beta^L_{nm}
\bH^L_{nm}  (\bx) + \sum_{n=1}^\infty \frac{\iota\K_S}{n(n+1)c_S^2} \sum_{m=-n}^n \beta^M_{nm} \bH^M_{nm}  (\bx)
\nonumber
\\
&+
 \sum_{n=1}^\infty \frac{\iota\K_S}{n(n+1)c_S^2} \sum_{m=-n}^n  \beta^N_{nm} \bH^N_{nm} (\bx),
\end{align*}
where
\begin{align*}
&\beta^\jmath_{nm}=\int_{\partial D} \overline{\bJ^\jmath_{nm} (\by)} \cdot \bpsi(\by) d\sigma(\by).
\end{align*}
Here, $(\bphi,\bpsi)$ is the solution of the system  of integral equations \eqref{integral-system}.

We are now fully equipped to announce the ESC of an inclusion in three-dimensions. We have the following definition.

\begin{definition}\label{def. scattering coeffd}

Let $(\bphi_{kl}^\jmath,\bpsi_{kl}^\jmath)$ be the solution of the integral system (\ref{integral-system}) when $\bu^{\rm inc}=\bJ^\jmath_{kl}$, for all integers $k\geq k_\jmath$ and $l=-k,\cdots, k$. Then, the ESC, denoted by  $W_{(n,m)(k,l)}^{\jmath,\jmath'}$, for all $n\geq n_\jmath$ and $m=-n,\cdots, n$,  associated to the compression modulus $\lambda$, shear modulus $\mu$, density \(\rho\), and frequency \(\omega\), is defined by
\begin{align*}
W_{(n,m)(k,l)}^{\jmath,\jmath'}[\lambda, \mu, \rho,\omega]=\int_{\partial D} \overline{\bJ^\jmath_{nm}(\by)}\cdot \bpsi_{kl}^{\jmath'} (\by)d\sigma(\by), \quad \forall\,\jmath,\jmath'=L,M,N.
\end{align*}

\end{definition}

The following result on the decay rate of the ESC holds.

\begin{lemma}\label{estimate on scattering coefficient}

For each $\jmath,\jmath'=L,M,N$, there exist a constant $C_{\jmath,\jmath'}$   depending only on the material parameters $(\lambda, \mu, \rho,\omega)$ such that
\begin{align*}
\left| W_{(n,m)(k,l)}^{\jmath,\jmath'} [\lambda, \mu, \rho,\omega]\right|\leq \frac{C_{\jmath,\jmath'}^{n+k-2}}{n^{n-1} k^{k-1}},
\end{align*}
for all $n\geq n_\jmath$, $m=-n,\cdots, n$, $k\geq k_{\jmath'}$ and $l=-k,\cdots, k$.
\end{lemma}

\begin{proof}

The proof is very similar to that of Lemma 3.1 in \cite{ammari2013enhancementMaxwell}. For the sake of completeness, we briefly sketch the proof below. First, recall that, by the Stirling's formula, we have
$$
k! = \sqrt{2\pi k}(k/e)k(1 + O(1)), \quad\text{as }\, k\to+\infty.
$$
Therefore, there exists a constant $C>0$ independent of $k$ such that, when $k\to+\infty$ and $t$ is fixed,
\begin{align}
j_k(t)= \frac{t^k}{1\cdot 3\cdots(2k+1)}\left(1+O\left(\frac{1}{k}\right)\right)=O\left(\frac{C^k t^k}{k^k}\right),\label{jEstimate}
\end{align}
uniformly on compact subsets of $\RR$. Moreover, by the recurrence formula (see, e.g., \cite[Formula 10.51.2]{NIST}),
\begin{align}
j_k'(t)=-j_{k+1}(t)+\frac{k}{t}j_k(t),
\end{align}
it is easy to get
\begin{align}
j_k'(t)=O\left(\frac{C^kt^{k-1}}{k^{k-1}}\right), \quad\text{as }\, k\to +\infty,
\label{jPrimeEstimate}
\end{align}
where the constant $C$ is independent of $k$. Consequently, by the definition of functions $\bJ^\jmath_{kl}(\bx)$, Theorem \ref{thmSolve}, and estimates \eqref{jEstimate}-\eqref{jPrimeEstimate}, we have
\begin{align*}
&\left\|\bJ^\jmath_{nm}\right\|_{L^2(\partial D)^3}\leq \left(\frac{C_1^\jmath}{n}\right)^{n-1},
\\
&\left\|\bpsi_{kl}^{\jmath'}\right\|_{L^2(\partial D)^3}\leq C\left(\left\|\bJ^{\jmath'}_{kl}\right\|_{H^1(\partial D)^3}+\left\|\bT\left[\bJ^{\jmath'}_{kl}\right]\right\|_{L^2(\partial D)^3}\right)
\leq \left(\frac{C_2^{\jmath'}}{k}\right)^{k-1},
\end{align*}
for some constants $C_1^\jmath$ and $C_2^{\jmath'}$ dependent on material parameters and frequency $\omega$ but independent of $n$ and $k$. Finally, the proof is completed by substituting the estimates for the norms of $\bJ^\jmath_{nm}$ and $\bpsi_{kl}^{\jmath'}$  in the definition of the ESC and choosing $C_{\jmath,\jmath'}$  appropriately in terms of $C_1^\jmath$ and $C_2^{\jmath'}$.
\end{proof}

We conclude this subsection with Theorem \ref{symmetry} on the symmetry of the ESC in three-dimensions. In fact,  this is the statement of the \emph{reciprocity principle} for the elastic fields in terms of the ESC. The reciprocity principle in elastic media refers to the link between the far-field amplitudes of the scattered fields in two reciprocal configurations: (a) when the scattered field is radiated by a source with a specific incidence direction and is observed along another direction, (b) when the positions of the source and the observer are swapped and the orientation of all the momenta is reversed. We refer the interested readers to consult \cite{ Dassios87, Varath2, Varath} for further discussion on the reciprocity in elastic media. The proof of Theorem \ref{symmetry} is furnished in Appendix \ref{app:sym}.

\begin{theorem}\label{symmetry}
For all integers $n\geq n_\jmath$, $k\geq k_{\jmath'}$, $m=-n,\cdots, n$,
$k=-l,\cdots, l$, and $\jmath,\jmath'=L,M,N$, we have
\begin{align*}
W^{\jmath,\jmath'}_{(n,m),(k,l)}[D]=\overline{W^{\jmath',\jmath}_{(k,l),(n,m)}[D]}.
\end{align*}
\end{theorem}

\subsection{Far-Field Amplitudes and Elastic Scattering Coefficients}

Consider a general incident field of the form
\begin{align}
\bu^{\rm inc}=&
\ds\sum_{k=1}^{\infty}\frac{1}{\sqrt{k(k+1)}}\sum_{l=-k}^{k}
\Bigg(
b_{kl}\bJ^M_{kl}+c_{kl}\bJ^N_{kl}\Bigg)
+\ds\sum_{k=0}^{\infty}\sum_{l=-k}^{k}
a_{kl}\bJ^L_{kl},
\label{Gen-Uinc}
\end{align}
for some constants $a_{kl}$,  $b_{kl}$, and $c_{kl}$. Then, by the superposition principle, the solution $\bpsi$  to the integral system  \eqref{integral-system} corresponding to $\bu^{\rm inc}$ in \eqref{Gen-Uinc}  is given by
\begin{equation}
\label{varphi}
\bpsi=\sum_{k=1}^{\infty}\frac{1}{\sqrt{k(k+1)}}\sum_{l=-k}^{k}
\left( b_{kl}\bpsi_{kl}^M+ c_{kl}\bpsi_{kl}^N\right)
+
\sum_{k=0}^{\infty}\sum_{l=-k}^{k} a_{kl} \bpsi_{kl}^L.
\end{equation}
Substituting \eqref{varphi} into \eqref{u-int-rep} and using Definition \ref{def. scattering coeffd}, one can represent the scattered field $\bu^{\rm sc}$ as
\begin{align}
\bu^{\rm sc}(\bx)
=& \sum_{n=0}^\infty \frac{\iota\K_P}{c_P^2} \sum_{m=-n}^n \gamma^L_{nm} \bH^L_{nm}  (\bx)+ \sum_{n=1}^\infty \frac{\iota\K_S}{n(n+1) c_S^2} \sum_{m=-n}^n \gamma^M_{nm} \bH^M_{nm}  (\bx)
\nonumber
\\
&+ \sum_{n=1}^\infty \frac{\iota\K_S}{ n(n+1)c_S^2} \sum_{m=-n}^n  \gamma^N_{nm} \bH^N_{nm} (\bx), \qquad\text{as }\,|\bx|\to+ \infty,
\label{scattered wave in term of scattering coeffd}
\end{align}
where, for all $n\geq n_\jmath$, $m=-n,\cdots,n$, and $\jmath=L,M,N$,
\begin{align}
\gamma^\jmath_{nm} = &\sum_{k=0}^\infty \sum_{l=-k}^k
a_{kl} W_{(n,m)(k,l)}^{\jmath, L}
+\sum_{k=1}^\infty\frac{1}{\sqrt{k(k+1)}} \sum_{l=-k}^k \Bigg(b_{kl}W_{(n,m)(k,l)}^{\jmath, M} +c_{kl} W_{(n,m)(k,l)}^{\jmath, N}
\Bigg).
\label{alpha, beta gamma}
\end{align}

On the other hand, the Kupradze radiation condition guarantees the existence of two analytic functions $\bu^\infty_P,\bu^\infty_S:\mathbb{S}^2\to \CC^3$, respectively called the longitudinal and transverse far-field patterns or far-field scattering amplitudes, such that
\begin{align}
\bu^{\rm sc}(\bx)=\frac{e^{\iota\K_P|\bx|}}{\K_P|\bx|}\bu^{\infty}_{P}[\lambda,\mu,\rho,\omega](\hat{\bx})+\frac{e^{\iota\K_S|\bx|}}{\K_S|\bx|}\bu^{\infty}_{S}[\lambda,\mu,\rho,\omega](\hat{\bx})+O\left(\frac{1}{{|\bx|}^2} \right),
\label{far-field pattern}
\end{align}
as $|\bx|\to+\infty$.
In fact,  it can be seen  from \eqref{interior multipole fields} and \eqref{behave of sph Bf} that the multipole fields behave as (see, e.g., \cite{pao1978betti})
\begin{align*}
&{\bH^L}_{nm}(\bx)  \thicksim \frac{e^{\iota\K_P|\bx|}}{\K_P|\bx|}\left(e^{-\frac{\iota n\pi}{2}}\bA_{nm}\right),
\\
&
\bH^M_{nm}(\bx) \thicksim\frac{e^{\iota\K_S|\bx|}}{\K_S|\bx|}\left(e^{-\frac{\iota(n+1)\pi}{2}}\sqrt{n(n+1)}\bC_{nm}\right),
\\
&\bH^N_{nm}(\bx)\thicksim\frac{e^{\iota\K_S|\bx|}}{\K_S|\bx|}\left(e^{-\frac{\iota n\pi}{2}}\sqrt{n(n+1)}\bB_{nm}\right),
\end{align*}
since the spherical Bessel functions $h_n^{(1)}$ and $(h^{(1)}_{n})^{\prime}$ behave like
\begin{equation}
\label{behave of sph Bf}
(h^{(1)}_{n})(t)\thicksim\frac{1}{t}e^{\iota t}e^{-\frac{\iota(n+1)\pi}{2}}
\quad\text{and}\quad
\left(h^{(1)}_{n}\right)^{\prime}(t)\thicksim\frac{1}{t}e^{\iota t}e^{-\frac{\iota n\pi}{2}},
\end{equation}
as $t\to+\infty$ (see, e.g., \cite[Eq. 2.41]{colton2003inverse}). Here, $f(t)\thicksim g(t)$ indicates that the terms of order $O(t^{-2})$ are neglected, i.e, $f(t)=g(t)+O(t^{-2})$.
Hence, the following result holds.
\begin{proposition}\label{PropFarF}
	If $\bu^{\rm inc}$ is of the form \eqref{Gen-Uinc} then the corresponding scattering amplitudes of the scattered field can be represented as
\begin{align}
\bu^\infty_{P}[\lambda,\mu,\rho,\omega](\hat{\bx}) =& \sum_{n=0}^{\infty}\frac{\iota\K_P}{c^2_P}\sum_{m=-n}^{n} \gamma^L_{nm}\left(e^{-\frac{\iota n\pi}{2}}\bA_{nm}\right),
\label{uPinfty}
\\
\bu^\infty_{S}[\lambda,\mu,\rho,\omega](\hat{\bx}) =&\sum_{n=1}^{\infty}\frac{\iota\K_S}{ \sqrt{n(n+1)} c^2_S}\sum_{m=-n}^{n}\Bigg[ \gamma^N_{nm}\left(e^{-\frac{\iota(n+1)\pi}{2}}\bC_{nm}\right)
+\gamma^M_{nm} \left(e^{-\frac{\iota n\pi}{2}}\bB_{nm} \right)\Bigg],\label{uSinfty}
\end{align}
where  the constants $\gamma^\jmath_{nm}$, for $\jmath=L,M,N$,  are given in \eqref{alpha, beta gamma}.
\end{proposition}
\begin{remark}
The following remarks are in order. Since $\{ \bA_{nm},\bB_{nm},\bC_{nm} \}$ is an orthonormal set in the inner-product space, one can estimate the near-field scattering signature of the inclusion $D$ from its far-field scattering amplitudes by using \eqref{scattered wave in term of scattering coeffd} and  \eqref{uPinfty}-\eqref{uSinfty}.  Further, if one can calculate the scattering coefficients of the inclusion $D$ then the scattering amplitudes $\bu^{\infty}_P$ and $\bu^{\infty}_S$ can also be calculated. For example,  if we consider the plane shear incident wave  $\bu^{\rm inc}=\bq e^{\iota\K_S\bx\cdot \bd}$ as in \eqref{incident s-wave} then
$$
a_{kl}=0,\quad b_{kl}=-4\pi \iota^k (\overline{\bC_{kl}(\bd)}\cdot \bq),\quad c_{kl}=-4\pi \iota^{k+1} (\overline{\bB_{kl}(\bd)}\cdot \bq).
$$
Hence, the scattering amplitudes $\bu^\infty_P$  and $\bu_S^\infty$ defined in \eqref{uPinfty}-\eqref{uSinfty}, respectively,  admit the scattering coefficients
\begin{align*}
\gamma^\jmath_{nm} = -\sum_{k=1}^\infty\frac{4\pi \iota^k}{\sqrt{k(k+1)}} \sum_{l=-k}^k \Bigg((\overline{\bC_{kl}(\bd)}\cdot \bq)&W_{(n,m)(k,l)}^{\jmath, M}+\iota (\overline{\bB_{kl}(\bd)}\cdot \bq)W_{(n,m)(k,l)}^{\jmath, N}
\Bigg).
\end{align*}
\end{remark}

\section{Three-dimensional ESC-Vanishing Structures}\label{sect:SVS}

In this section, we  design concentric multi-layered spherical coatings for an elastic object so that its ESC vanish at a fixed frequency in order to construct an effective elastic cloaking device for rendering the objects inside the core \emph{invisible}. These structures are coined as the \emph{ESC-vanishing structures} \cite{abbas2016elastic, ammari2012enhancement, ammari2013enhancement, ammari2013enhancementHelmoltz, ammari2013enhancementMaxwell}.

\subsection{Design of ESC-Vanishing Structure}
We first choose an integer $\L\in\NN$ as a parameter for the desired number of layers and pick $\L+1$ real numbers, $r_1, r_2,\cdots , r_{\L+1}\in\RR_+$, such that  $2=r_1>r_2>\cdots>r_{\L+1}=1$. Then, we construct a sequence of layers,  $A_0, A_1,\cdots, A_{\L+1}\subset\RR^3$,  by
\begin{align}
A_\ell=
\begin{cases}
\RR^3\setminus \overline{\bigcup^{\L+1}_{\ell=1}A_\ell}, & \text{ for }\ \ell=0,
\\
\{\bx: r_{\ell+1}\leq |\bx| < r_\ell\},  & \text{ for }\ \ell=1, \cdots ,\L,
\\
\{\bx:|\bx|<1\},& \text{ for }\ \ell=\L+1,
\end{cases}
\label{A_ell}
\end{align}
and denote the interfaces between the adjacent layers by $\Gamma_\ell$. Precisely,
\begin{align*}
\Gamma_\ell:=\{|\bx|=r_\ell\},\quad \ell=1,\cdots, \L+1.
\end{align*}
Let $A_\ell$, for $\ell=1,\cdots, \L+1$, be loaded with an elastic material with the pair of compression and shear moduli $(\lambda_\ell,\mu_\ell)$ and density $\rho_\ell$. Accordingly, we define piece-wise constant material parameters of the entire layered elastic formation by
\begin{align}
\left[\lambda;  \mu; \rho\right](\bx):=\sum_{\ell=0}^{\L+1}\left[\lambda_\ell;  \mu_\ell; \rho_\ell\right] \chi_{A_\ell}(\bx),\qquad\bx\in\RR^3.
\label{parameters}
\end{align}
Further, the scattering coefficients $W_{(n,m)(k,l)}^{\jmath,\jmath'}$, for $\jmath,\jmath'=L,M,N$, are defined in this section as in Definition \ref{def. scattering coeffd}.

We will employ the transmission conditions on each interface $\Gamma_\ell$, i.e., we impose
\begin{equation}
\bu^{\rm tot}\big|_+=\bu^{\rm tot}\big|_- \,\,\text{and} \,\,\bT_\ell[\bu^{\rm tot}]\big|_+=\bT_{\ell+1}[\bu^{\rm tot}]\big|_-, \quad \text{ at } \Gamma_\ell, \quad \ell =1,\cdots, \L,
\label{transimation condition in ch3}
\end{equation}
where
\begin{align*}
\bT_\ell [\bw]:= \lambda_\ell(\nabla\cdot\bw)\bnu+\mu_\ell(\nabla\bu+\nabla\bw^\top), \quad \ell=0,\cdots, \L+1,
\end{align*}
for an arbitrary smooth function $\bw$. It is assumed that the cloaked region, $A_{\L+1}$, is a cavity so that $\bu^{\rm tot}$ satisfies a traction-free boundary condition,
\begin{equation}
\label{traction on L+1}
\bT_{\L+1}[\bu^{\rm tot}]=\mathbf{0}, \qquad \text{on} \ \Gamma_{\L+1}=\partial A_{\L+1}=\partial\mathbb{S}^2.
\end{equation}

We  have the following definition of a three-dimensional ESC-vanishing structure.
\begin{definition}\label{s-vStructure}
The multi-layered medium \eqref{A_ell} together with parameters $(\lambda,\mu,\rho)$ defined in \eqref{parameters} is called the ESC-vanishing structure at frequency $\omega$ if $W_{(n,m)(k,l)}^{\jmath,\jmath'}=0$ for all $\jmath,\jmath'=L,M,N$,   integers $n\geq n_\jmath$, $k\geq k_{\jmath'}$, $m=-n,\cdots, n$, and $l=-k,\cdots,k$.
\end{definition}

In order to design the ESC-vanishing structure, it suffices to find appropriate $(\lambda,\mu,\rho)$  such that $W_{(n,m)(k,l)}^{\jmath,\jmath'}=0$. In the rest of this section, we establish constitutive relations in order to design parameters  $(\lambda,\mu,\rho)$ . Towards this end, we seek the solution $\bu_{nm}^{\rm tot}$ of Lam\'e system \eqref{elastodynamics equations} of the form
\begin{align*}
\bu_{nm}^{\rm tot}(\bx)=\bu_{nm}^{\rm inc}(\bx)+\bu_{nm}^{\rm sc}(\bx), \qquad n>0,\quad m=-n,\cdots, n,
\end{align*}
with
\begin{align*}
&\bu_{nm}^{\rm inc}{(\bx)}=
\sum_{\jmath=L,M,N} a^\jmath_{\ell}{\bJ^\jmath}_{nm}(\bx), &\bx\in A_\ell, \,\, \ell=0,\cdots,\L,
\\
& \bu_{nm}^{\rm sc}(\bx)=
\sum_{\jmath=L,M,N} b^\jmath_{\ell} {\bH^\jmath}_{nm}(\bx),   &\bx\in A_\ell, \, \ell=0,\cdots,\L,
\end{align*}
where the coefficients $a^\jmath_{\ell}, b^\jmath_{\ell}\in\CC$, for $\ell=0, \cdots, \L$,  are to be determined subject to
\begin{align*}
(a_0^L)^2+ (a_0^M)^2+(a_0^N)^2\neq 0\quad \text{and}\quad b_0^\jmath=0, \quad\text{for all }\,\jmath=L,M,N.
\end{align*}
In fact, the coefficients $a_0^\jmath$ control the proportion of $\bJ^\jmath_{nm}$ in the  incident field  whereas $b_0^\jmath$ control the proportion of $\bH^\jmath$ in the scattered field. By comparison with multipolar expansion \eqref{scattered wave in term of scattering coeffd}-\eqref{alpha, beta gamma}, it is evident that
\begin{align}
W_{(n,m)(k,l)}^{\jmath,\jmath'}= \zeta^{\jmath}_0\,b^\jmath_0 \quad\text{when}\,\, a_0^{\jmath''}=\delta_{\jmath'\jmath''},
\qquad \forall \jmath, \jmath', \jmath''=L,M,N.\label{WB-relation}
\end{align}
Here,
\begin{align*}
\zeta^L_0:=
-\frac{\iota c_{P,\L+1}^2}{\K_{P,\L+1}},
\quad\text{and}\quad
\zeta^N_0=\zeta^M_0:=-\frac{\iota n(n+1)c_{S,\L+1}^2}{\K_{S,\L+1}},
\end{align*}
with
\begin{align*}
c_{P,\ell}=\sqrt{\frac{\lambda_\ell+2\mu_\ell}{\rho_\ell}},
\quad
c_{S,\ell}=\sqrt{\frac{\mu_\ell}{\rho_\ell}},
\quad
\K_{\alpha,\ell}=\frac{\omega}{c_{\alpha,\ell}}, \qquad\forall \ell=0,\cdots, \L+1.
\end{align*}

\subsection{Determination of Coefficients $a_\ell^\jmath$ and $b_\ell^\jmath$}\label{ss:Coefficients}

In order to determine the coefficients  $a_\ell^\jmath$ and $b_\ell^\jmath$, we require six equations for each $\ell=1,\cdots, \L+1$. Towards this end, we will effectively use the transmission conditions \eqref{transimation condition in ch3} and the traction-free boundary condition \eqref{traction on L+1}.

We first use the transmission condition in \eqref{transimation condition in ch3} for the Dirichlet data to get
\begin{align}
&\sum_{\jmath=L,M,N} \Big[a^\jmath_{\ell-1}{[\bJ^\jmath_{nm}]}_{\ell-1}
+b^\jmath_{\ell-1}{[\bH^\jmath_{nm}]}_{\ell-1}\Big]
 \nonumber
 \\
 &\qquad\qquad
 =\sum_{\jmath=L,M,N} \Big[a^\jmath_{\ell}[\bJ^\jmath_{nm}]_{\ell}
+b^\jmath_{\ell}[\bH^\jmath_{nm}]_{\ell}\Big],
\quad \text{on }\,\Gamma_\ell \, \text{ for }\,\ell=1,\cdots,\L,
\label{general trans cond}
\end{align}
where the notation $[\bw]_\ell$ indicates that the quantity $\bw$ is associated to the parameters $(\lambda_{\ell},\mu_{\ell}, \rho_{\ell})$ and to the corresponding wave-numbers $\K_{\alpha, \ell}$. Substituting the expressions \eqref{interior multipole fields}-\eqref{exterior multipole fields} for the vector potentials $\bJ^\jmath_{nm}$ and $\bH^\jmath_{nm}$ into \eqref{general trans cond}, we get
\begin{align}
[E_{n}]_{\ell-1}(r_\ell)\bA_{nm}&+[F_{n}]_{\ell-1}(r_\ell){\bB_{nm}}+[G_{n}]_{\ell-1}(r_\ell)\bC_{nm}
 \nonumber
 \\
= &[E_{n}]_{\ell}(r_\ell)\bA_{nm}+[F_{n}]_{\ell}(r_\ell)\bB_{nm}+[G_{n}]_{\ell}(r_\ell)\bC_{nm},
\label{expand trans cond}
\end{align}
for all $\ell=1,\cdots,\L$, with
\begin{align}
[E_{n}]_{\ell}(r)= & a^L_{\ell}(j_n)^\prime(\K_{P,\ell}r)+a^N_{\ell}\frac{n(n+1)}{\K_{S,\ell}r} j_n(\K_{S,\ell}r)+b^L_{\ell}(h_n^{(1)})^\prime(\K_{P,\ell}r)
+b^N_{\ell}\frac{n(n+1)}{\K_{S,\ell}r}h_n^{(1)}(\K_{S,\ell}r),
\label{F}
\\
[F_{n}]_{\ell}(r) =&
\Bigg[a^L_{\ell}\frac{j_n(\K_{P,\ell}r)}{\K_{P,\ell}r}+a^N_{\ell}\frac{\mathcal{J}_n(\K_{S,\ell}r)}{\K_{S,\ell}r}+b^L_{\ell}\frac{h_n^{(1)}(\K_{P,\ell}r)}{\K_{P,\ell}r}
+b^N_{\ell}\frac{\mathcal{H}_n(\K_{S,\ell}r)}{\K_{S,\ell}r}\Bigg]\sqrt{n(n+1)},
\label{G}
\\
[G_{n}]_{\ell}(r) =& \Big[a^M_{\ell}j_n({\K_{S,\ell}r})+b^M_{\ell}h_n^{(1)}(\K_{S,\ell}r)\Big]\sqrt{n(n+1)}.
\label{E}
\end{align}
Therefore, multiplying \eqref{expand trans cond} with $\overline{\bA_{kl}}$, $\overline{\bB_{kl}}$, and $\overline{\bC_{kl}}$, respectively, to form a dot product and then applying orthogonality conditions \eqref{othogonal condition 1}-\eqref{othogonal condition 2}, one gets three recursive relations,
\begin{align*}
& [E_{n}]_{\ell-1}=[E_{n}]_{\ell},
\quad
[F_{n}]_{\ell-1}={[F_{n}]}_{\ell},
\quad
[G_{n}]_{\ell-1}=[G_{n}]_{\ell}, \quad\text{on }\,\Gamma_\ell,
\end{align*}
for $n\in\NN$, $m=-n,\cdots, n$, and $\ell=1,\cdots, \L$.  In the same spirit, we form another set of three equations using the transmission condition \eqref{transimation condition in ch3} for the surface traction and subsequently invoking the orthogonality conditions \eqref{othogonal condition 1}-\eqref{othogonal condition 2}.  The interested readers are referred to Appendix \ref{app:tractions} for the expressions of the surface traction for different multipole elastic fields.

The aforementioned procedure renders six simultaneous equations in six unknowns  that can be written in a matrix form as
\begin{equation}
\begin{bmatrix}
[\bP_{n}^{L,N}] _\ell(r_\ell)   & \mathbf{0}_{4\times 2}
\\\\
\mathbf{0}_{2\times 4}      &  [\bP_{n}^{M}]_\ell(r_\ell)
\end{bmatrix}
\begin{bmatrix}
a^L_{\ell}\\a^N_{\ell}\\b^L_{\ell}\\b^N_{\ell}\\a^M_{\ell}\\b^M_{\ell}
\end{bmatrix}
=
\begin{bmatrix}
[\bP_{n}^{L,N}]_{\ell-1}(r_\ell)     & \mathbf{0}_{4\times 2}
\\\\
\mathbf{0}_{2\times 4}                &  [\bP_{n}^{M}]_{\ell-1}(r_\ell)
\end{bmatrix}
\begin{bmatrix}
a^L_{\ell-1}\\a^N_{\ell-1}\\ b^L_{\ell-1}\\b^N_{\ell-1}
\\ a^M_{\ell-1}\\b^M_{\ell-1}
\end{bmatrix},
\label{transmission matrix}	
\end{equation}
where ${[\bP_{n}^{L,N}]} _\ell(r_\ell):={[\bP_{n}^{L,N}]}_\ell(r_\ell;\lambda,\mu,\rho,\omega)$ is a $4\times 4$ sub-matrix corresponding to the fields $(\bJ^L_{nm},\bJ^N_{nm},\bH^L_{nm},\bH^N_{nm})$, whereas $[\bP_{n}^{M}]_\ell(r_\ell):=[\bP_{n}^{M}]_\ell(r_\ell;\lambda,\mu,\rho,\omega)$ is a $2\times 2$  sub-matrix corresponding to  $(\bJ^M_{nm},\bH^M_{nm})$.  The expressions of the elements of  ${[\bP_{n}^{L,N}]} _\ell$ and ${[\bP_{n}^{M}]} _\ell$  can be found in \ref{app:Ps}.  It is worth mentioning that the sub-matrices $[\bP_{n}^{L,N }]_\ell$  and $[\bP_{n}^M]_\ell $ are invertible. Therefore, matrix equation \eqref{transmission matrix} yields a recursive relation
\begin{equation}
\begin{bmatrix}
a^L_\ell\\a^N_\ell\\b^L_\ell\\b^N_\ell\\a^M_\ell\\b^M_\ell
\end{bmatrix}
=
\begin{bmatrix}
[\bP_{n}^{L,N}] _\ell^{-1}[\bP_{n}^{L,N}] _{\ell-1}  & \mathbf{0}_{4\times 2}
\\\\
\mathbf{0}_{2\times 4}   &  [\bP_{n}^{M}]_\ell^{-1}[\bP_{n}^{M}] _{\ell-1}
\end{bmatrix}
\begin{bmatrix}
a^L_{\ell-1}\\a^N_{\ell-1}\\b^L_{\ell-1}\\b^N_{\ell-1}\\a^M_{\ell-1}\\b^M_{\ell-1}
\end{bmatrix},
\label{change subject}
\end{equation}
for all $\ell=1,\cdots \L$.

In order to solve the recursive relation \eqref{change subject}, we invoke the zero-traction condition \eqref{traction on L+1} on $\Gamma_{\L+1}(i.e., r_{\L+1}=1)$. This furnishes
\begin{equation}
\begin{bmatrix}
0\\0\\0\\0\\0\\0
\end{bmatrix}=
\begin{bmatrix}
{[\bQ_{n}^{L,N}]}_{\L} & \mathbf{0}_{4\times 2}
\\\\
\mathbf{0}_{2\times 4}  &  {[\bQ_{n}^{M}]}_{\L}
\end{bmatrix}
\begin{bmatrix}
a^L_{\L}\\a^N_{\L}\\b^L_{\L}\\b^N_{\L}\\a^M_{\L}\\b^M_{\L}
\end{bmatrix},
\label{boundary G}
\end{equation}
where $[\bQ_{n}^{L,N}]_{\L}$ and  $[\bQ_{n}^{M}]_{\L}$ are  $4\times 4$ and $2\times 2$ sub-matrices, respectively and the superscripts indicate their dependence on different wave-modes. The expressions of the elements of these matrices are also provided in Appendix \ref{app:Ps}.

Substituting \eqref{change subject} into \eqref{boundary G}, one arrives at
\begin{equation}
\label{All H}
\begin{bmatrix}
0\\0\\0\\0\\0\\0
\end{bmatrix}=
\begin{bmatrix}
\bR_{n}^{L,N} & \mathbf{0}_{4\times 2}
\\\\
\mathbf{0}_{2\times 4}  &  \bR_{n}^{M}
\end{bmatrix}
\begin{bmatrix}
a^L_0\\a^N_0\\b^L_0\\b^N_0\\a^M_0\\b^M_0
\end{bmatrix},
\end{equation}
where
\begin{align}
\begin{cases}
\bR_{n}^{L,N}[\lambda,\mu,\rho,\omega]:=
\left({(R_{n}^{L,N})}_{p,q}\right)=\ds
 [\bQ_{n}^{L,N}]_{\L} \prod_{\ell=1}^{{\L}} [\bP_{n}^{L,N}]_\ell^{(-1)}{[\bP_{n}^{L,N}]}_{\ell-1},
\\
\bR_{n}^{M}[\lambda,\mu,\rho,\omega]:=
\left({(R_{n}^{M})}_{p,q}\right)=\ds
[\bQ_{n}^{M}]_{\L}\prod_{\ell=1}^{\L} [\bP_{n}^{M}]_\ell^{(-1)}[\bP_{n}^{M}]_{\ell-1}.
\end{cases}\label{H^M}
\end{align}

\subsection{Constitutive Equations for Material Parameters}\label{ss:Constitutive}

If we consider an incident pressure wave $\bJ^L_{nm}$ for $n>0$ (i.e., $a_0^L=1$, $a_0^M=0$, and $a_0^N=0$) then  \eqref{All H} renders
\begin{align}
\begin{cases}
\ds(R^{L,N}_{n})_{1,1}+b^L_0 (R^{L,N}_{n})_{1,3} + b^N_0(R^{L,N}_{n})_{1,4} &=0,
\\\
\ds (R^{L,N}_{n})_{2,1}+b^L_0 (R^{L,N}_{n})_{2,3} + b^N_0(R^{L,N}_{n})_{2,4} &=0,
\\
\ds b^M_0 &=0.
\end{cases}\label{sys1}
\end{align}
Solving system of equations \eqref{sys1}, we get
\begin{equation}
\label{scattering of pure L wave H}
\begin{cases}
b^{LL}_0:=b^L_0=\ds\frac{(R^{L,N}_{n})_{1,1}(R^{L,N}_{n})_{2,4}-(R^{L,N}_{n})_{1,4}(R^{L,N}_{n})_{2,1}}{(R^{L,N}_{n})_{1,4}(R^{L,N}_{n})_{2,3}-(R^{L,N}_{n})_{1,3}(R^{L,N}_{n})_{2,4}},
\\
b^{NL}_0:=b^N_0=\ds\frac{(R^{L,N}_{n})_{1,3}(R^{L,N}_{n})_{2,1}-(R^{L,N}_{n})_{1,1}(R^{L,N}_{n})_{2,3}}{(R^{L,N}_{n})_{1,4}(R^{L,N}_{n})_{2,3}-(R^{L,N}_{n})_{1,3}(R^{L,N}_{n})_{2,4}},
\\
b^{ML}_0:=b^M_0=0.
\end{cases}
\end{equation}
In the above relations, we use the original coefficients $b_0^L$, $b_0^N$, and $b_0^M$ but with an explicit notation  $b_0^{LL}$, $b_0^{NL}$, and $b_0^{ML}$, in order to highlight both the incident and scattered wave-modes.  More specifically, the first superscript represents the type of the scattered wave whereas the second represents the type of the incident field.

Similarly, if we consider an incident $\bJ^N_{nm}$- wave for $n>0$ (i.e.,  $a_0^L=0,a_0^M=0$, and $a_0^N=1$) then we get system of equations from \eqref{All H} that will furnish
\begin{equation}
\label{scattering of pure N wave H}
\begin{cases}
b^{LN}_0=\ds\frac{(R^{L,N}_{n})_{1,2}(R^{L,N}_{n})_{2,4}-(R^{L,N}_{n})_{1,4}(R^{L,N}_{n})_{2,2}}{(R^{L,N}_{n})_{1,4}(R^{L,N}_{n})_{2,3}-(R^{L,N}_{n})_{1,3}(R^{L,N}_{n})_{2,4}},
\\
b^{NN}_0=\ds\frac{(R^{L,N}_{n})_{1,3}(R^{L,N}_{n})_{2,2}-(R^{L,N}_{n})_{1,2}(R^{L,N}_{n})_{2,3}}{(R^{L,N}_{n})_{1,4}(R^{L,N}_{n})_{2,3}-(R^{L,N}_{n})_{1,3}(R^{L,N}_{n})_{2,4}},
\\
b^{MN}_0=0.
\end{cases}
\end{equation}
Finally, when $\bJ^M_{nm}$-wave is incident for $n>0$ (i.e., $a_0^L=0$, $a_0^M=1$, and $a_0^N=0$), then
\begin{equation}
\label{scattering of pure M wave H}
b^{LM}_0=0,
\quad
b^{NM}_0=0,
\quad\text{and}\quad
b^{MM}_0=\ds \frac{{(R^{M}_{n})}_{1,1}}{{(R^{M}_{n})}_{1,2}}.
\end{equation}

\begin{remark}
The following remarks are in order.
\begin{enumerate}

\item The coefficients $b^{\jmath\jmath'}_0$ in  \eqref{scattering of pure L wave H}-\eqref{scattering of pure M wave H} indicate that a scattered $\bH^N_{nm}$-wave can emerge from an incidence of a $\bJ^L_{nm}$-wave and vise versa \textemdash the so-called mode-conversion phenomenon. However, this is not the case for $\bH^M_{nm}$, i.e., there is  no mode-conversion when the incident field only consists of $\bJ^M_{nm}$-wave.

\item Note that, the denominators, ${(R^{L,N}_{nm})_{1,4}(R^{L,N}_{nm})_{2,3}-(R^{L,N}_{nm})_{1,3}(R^{L,N}_{nm})_{2,4}}$ and $(R^{M}_{nm})_{1,2}$, in  \eqref{scattering of pure L wave H}-\eqref{scattering of pure M wave H} can not vanish. Otherwise, a contradiction to the uniqueness of the forward scattering problem can be derived  exactly as in the case of acoustic and electromagnetic problems discussed in \cite{ammari2013enhancementHelmoltz, ammari2013enhancementMaxwell}.

\item In the spherically layered structure defined by the sets $A_\ell$, the expressions of the ESC become much simpler than the general case. Specifically, thanks to \eqref{scattering of pure L wave H}-\eqref{scattering of pure M wave H}, relation \eqref{WB-relation} results in
\begin{align*}
W^{M,L}_{(n,m)(k,l)}=W^{M,N}_{(n,m)(k,l)}=W^{L,M}_{(n,m)(k,l)}=W^{N,M}_{(n,m)(k,l)}=0,\quad \forall n,m,k,l.
\end{align*}
Moreover, thanks to the  reciprocity result in Theorem \ref{symmetry}, and Eqs. \eqref{scattering of pure L wave H}-\eqref{scattering of pure M wave H},
\begin{align*}
W^{L,L}_{(n,m)(k,l)}=W^{N,N}_{(n,m)(k,l)}=W^{M,M}_{(n,m)(k,l)}=W^{N,L}_{(n,m)(k,l)}=W^{L,N}_{(n,m)(k,l)}=0,
\end{align*}
whenever $(n,m)\neq (k,l)$.
\item As the coefficients $b_0^{\jmath \jmath'}$ depend only on the integer $n$ and are completely independent of $m$. Therefore,
\begin{equation}
\label{define W_n}
\begin{cases}
W^{\jmath,\jmath'}_{(n,0)(n,0)}=W^{\jmath,\jmath'}_{(n,m)(k,l)}=:  W^{\jmath,\jmath'}_n,  & m= -n,\cdots,  n,\quad \jmath,\jmath'=L,N,
\\
W^{M,M}_{(n,0)(n,0)}=W^{M,M}_{(n,m)(k,l)}=: W^{M,M}_n, & m= -n,\cdots,  n.
\end{cases}
\end{equation}

\item Finally, the degenerate case $n=0$ can be dealt with analogously as the case $n>0$ discussed in Sections \ref{ss:Coefficients}-\ref{ss:Constitutive}. In this case, only a pressure wave $\bJ^L_{00}(\bx)=\left(j_0\right)'(\K_p r)\be_r/\sqrt{4\pi}$ will be incident. In fact, thanks to the spherical symmetry of the structure, we will look for solution $\bu^{\rm tot}_{00}$ of the Lam\'e system \eqref{elastodynamics equations} of the form
\begin{align*}
\bu_{00}^{\rm tot}(\bx)=\bu_{00}^{\rm inc}(\bx)+\bu_{00}^{\rm sc}(\bx),
\end{align*}
with
\begin{align*}
&\bu_{00}^{\rm inc}{(\bx)}=
a^L_{\ell}{\bJ^L}_{00}(\bx)
\quad \text{and}\quad \bu_{00}^{\rm sc}(\bx)=
 b^L_{\ell} {\bH^L}_{00}(\bx),   &\bx\in A_\ell, \, \ell=0,\cdots,\L,
\end{align*}
where $\bH^L_{00}(\bx)=\left(h^{(1)}_0\right)'(\K_p r)\be_r/\sqrt{4\pi}$ and the coefficients $a^L_{\ell}, b^L_{\ell}\in\CC$, for $\ell=0, \cdots, \L$,  are to be determined using transmission conditions subject to restrictions
\begin{align*}
a_0^L\neq 0\quad \text{and}\quad b_0^L=0.
\end{align*}
\end{enumerate}
\end{remark}

In view of the relationship \eqref{WB-relation} and the notation in \eqref{define W_n},  the scattering coefficients of the core $A_{\L+1}$ in the spherically layered structure are given by
\begin{align}
\begin{cases}
W^{L,L}_n
=\ds\zeta^L_0\left[\frac{(R^{L,N}_{n})_{1,1}(R^{L,N}_{n})_{2,4}-(R^{L,N}_{n})_{1,4}(R^{L,N}_{n})_{2,1}}{(R^{L,N}_{n})_{1,4}(R^{L,N}_{n})_{2,3}-(R^{L,N}_{n})_{1,3}(R^{L,N}_{n})_{2,4}}\right],
\\
W^{N,L}_n
=\ds\zeta^N_0\left[ \frac{(R^{L,N}_{n})_{1,3}(R^{L,N}_{n})_{2,1}-(R^{L,N}_{n})_{1,1}(R^{L,N}_{n})_{2,3}}{(R^{L,N}_{n})_{1,4}(R^{L,N}_{n})_{2,3}-(R^{L,N}_{n})_{1,3}(R^{L,N}_{n})_{2,4}}\right],
\\
W^{L,N}_n
=\ds\zeta^L_0\left[\frac{(R^{L,N}_{n})_{1,2}(R^{L,N}_{n})_{2,4}-(R^{L,N}_{n})_{1,4}(R^{L,N}_{n})_{2,2}}{(R^{L,N}_{n})_{1,4}(R^{L,N}_{n})_{2,3}-(R^{L,N}_{n})_{1,3}(R^{L,N}_{n})_{2,4}}\right],
\\
W^{N,N}_n
=\ds\zeta^N_0\left[ \frac{(R^{L,N}_{n})_{1,3}(R^{L,N}_{n})_{2,2}-(R^{L,N}_{n})_{1,2}(R^{L,N}_{n})_{2,3}}{(R^{L,N}_{n})_{1,4}(R^{L,N}_{n})_{2,3}-(R^{L,N}_{n})_{1,3}(R^{L,N}_{n})_{2,4}} \right],
\\
W^{M,M}_n
=\ds\zeta^M_0\left[\frac{(R^{M}_{n})_{1,1}}{(R^{M}_{n})_{1,2}} \right].
\end{cases}\label{W^{MM}}
\end{align}

Finally, from the expressions \eqref{W^{MM}} for the ESC, we conclude that in order to construct an ESC-vanishing structure, it suffices to look for the parameters $\lambda_\ell$, $\mu_\ell$, and $\rho_\ell$, for $\ell=1,\cdots, \L$, from the nonlinear algebraic equations
\begin{align*}
(R_n^{M})_{1,1}=0\quad\text{and}\quad (R_n^{L,N})_{p,q}=0, \quad \forall p,q=1,2,\quad n\in\NN.
\end{align*}

\subsection{Numerical examples}
In this section, we present several numerical examples to demonstrate the existence of ESC-vanishing structures, i.e., structures $(\lambda,\, \mu, \, \rho)$ of the form \eqref{parameters} satisfying \eqref{W^{MM}}. The parameters can be obtained by solving the following optimization problem
\begin{equation*}
  \min \mathfrak{J}(\lambda,\, \mu, \, \rho)=\sum_{n=0}^{T}\sum_{i,j=L,M,N} \left|W_n^{i,j}\right|^2, \quad T\in \mathbb{N},
\end{equation*}
where $T$ denotes the truncated order. Here we use the gradient descent method to identify the parameters. We emphasize that it is challenging to minimize the quantity for a large truncated order $T$. In the following examples, we take $T=2$, that is $n=0,1,2$. In addition, we let the frequency be $\omega=1$.
\begin{figure}
    \subfigure[]{\includegraphics[width=0.32\textwidth]{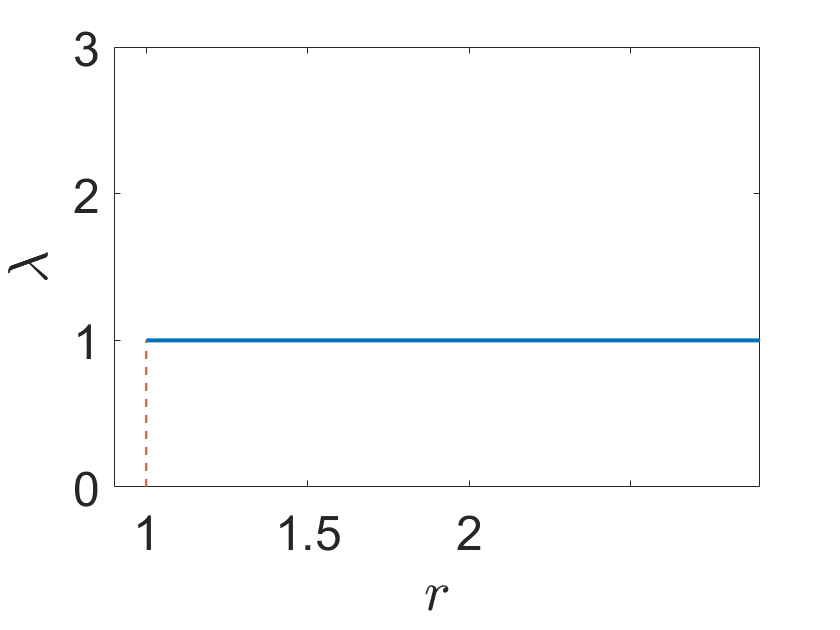}} \subfigure[]{\includegraphics[width=0.32\textwidth]{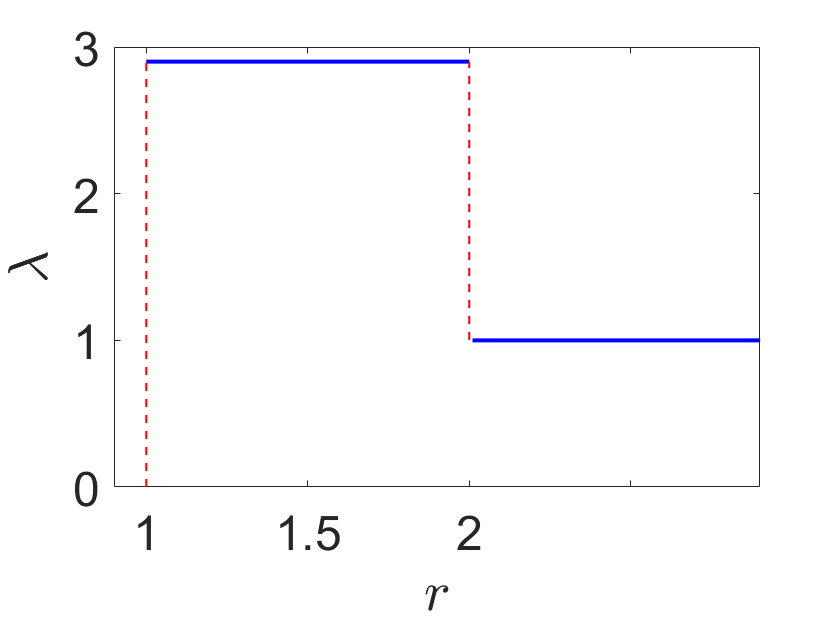}} \subfigure[]{\includegraphics[width=0.32\textwidth]{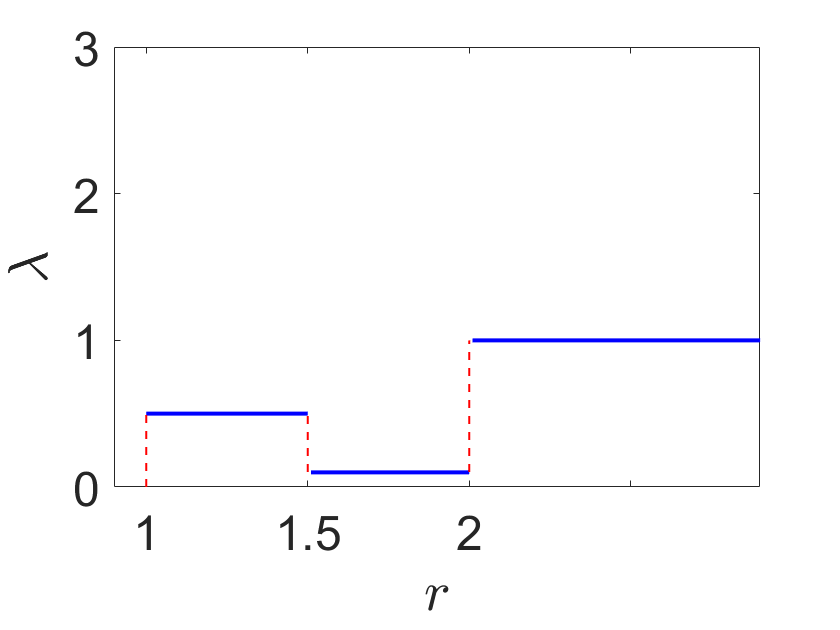}}\\
    \subfigure[]{\includegraphics[width=0.32\textwidth]{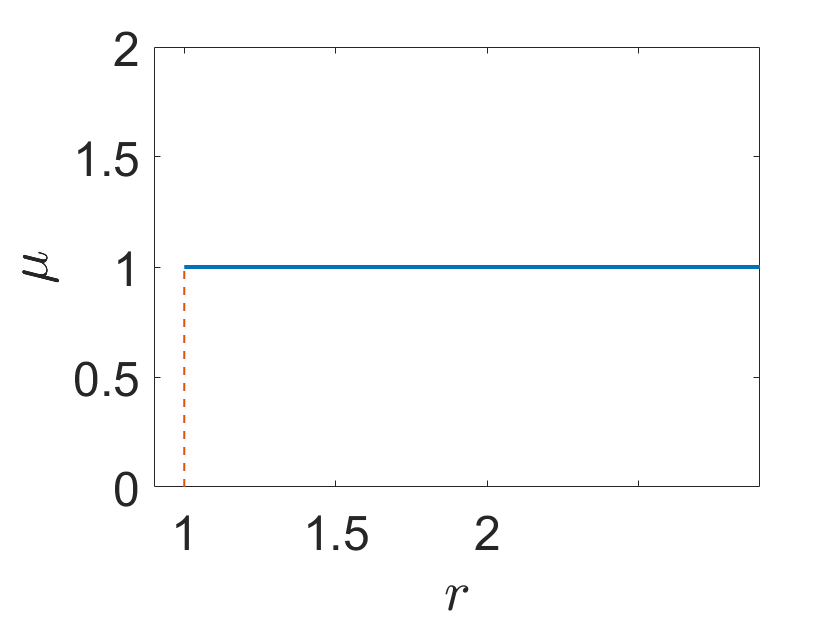}}    \subfigure[]{\includegraphics[width=0.32\textwidth]{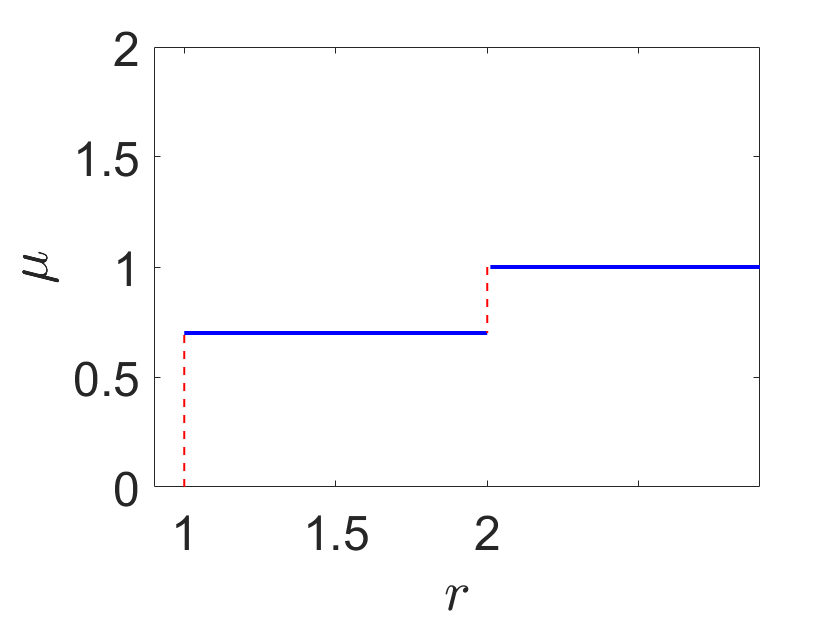}} \subfigure[]{\includegraphics[width=0.32\textwidth]{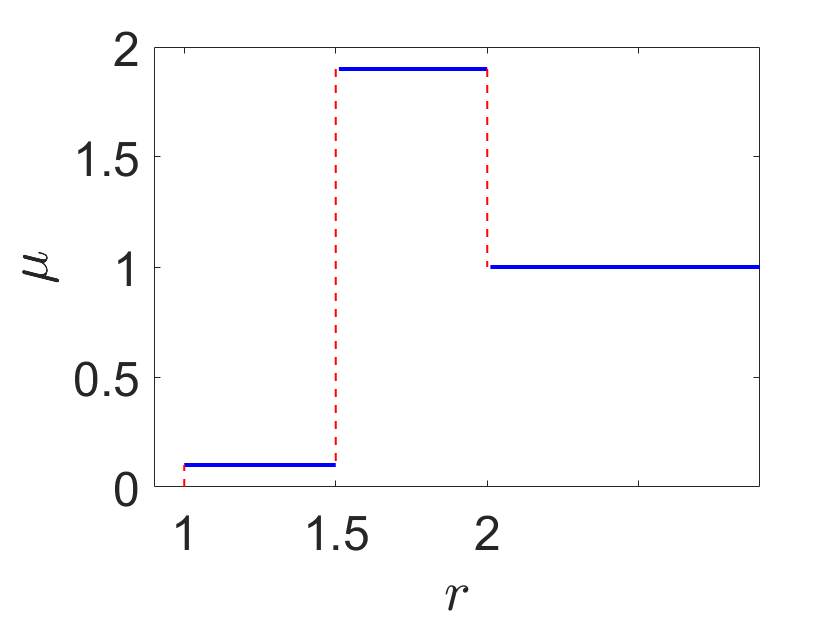}}\\
    \subfigure[]{\includegraphics[width=0.32\textwidth]{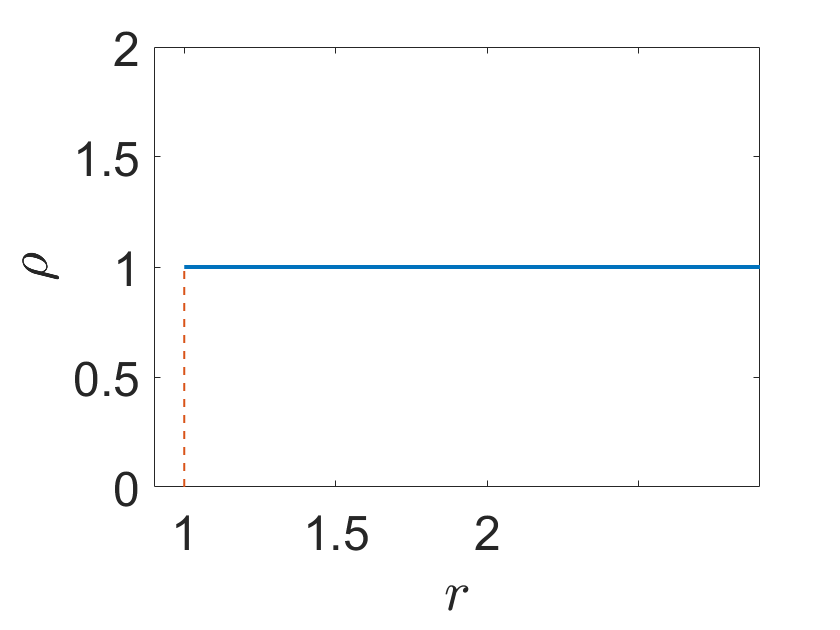}}
    \subfigure[]{\includegraphics[width=0.32\textwidth]{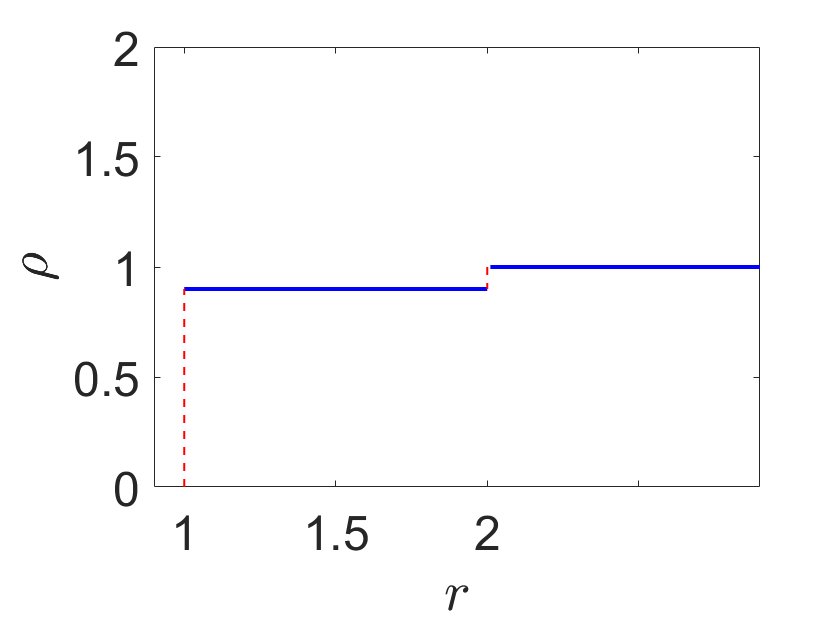}}
    \subfigure[]{\includegraphics[width=0.32\textwidth]{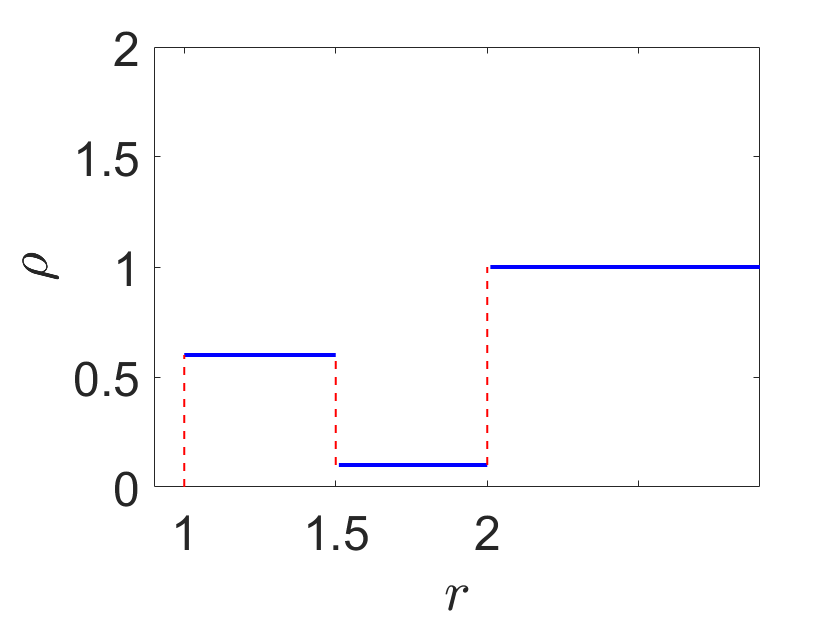}}
    \caption{The top row, center row and bottom row denote
    parameters $\lambda$, $\mu$ and $\rho$, respectively. The left column denotes the case with no layer, the center column denotes the case with $1$ cloaking layer, and the right column denotes the case with $2$ cloaking layers.
    } \label{fig:parameter}
\end{figure}

\begin{figure}
    \subfigure[0 layer]{\includegraphics[width=0.32 \textwidth]{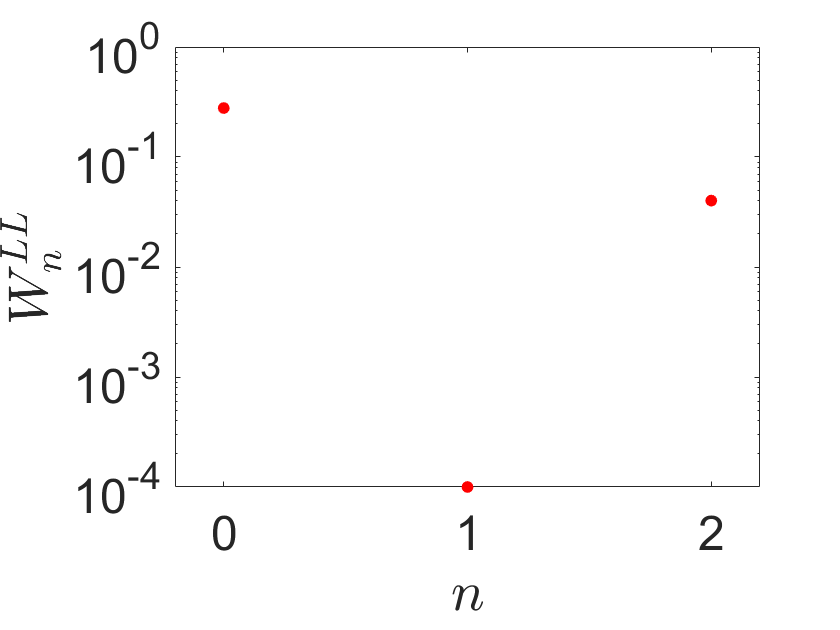}}    \subfigure[1 layer]{\includegraphics[width=0.32 \textwidth]{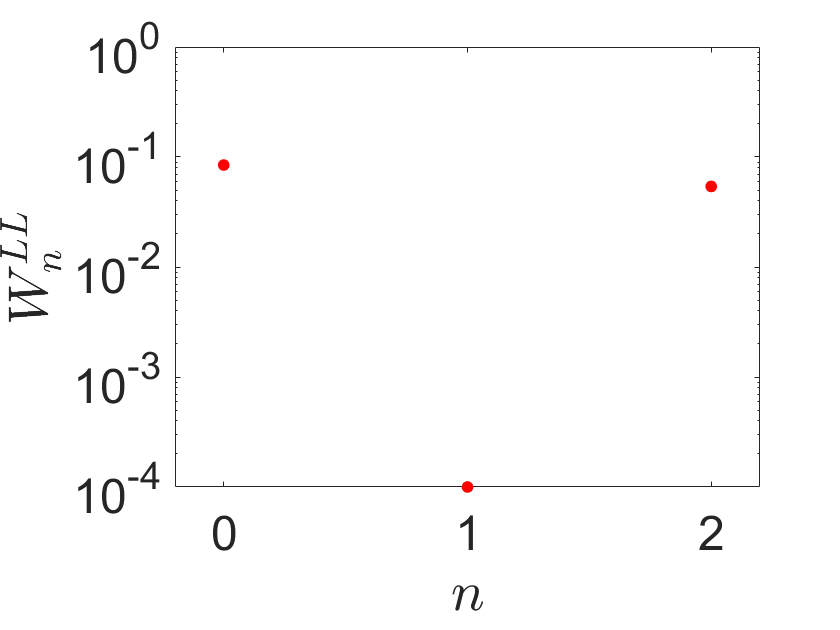}}    \subfigure[2 layers]{\includegraphics[width=0.32 \textwidth]{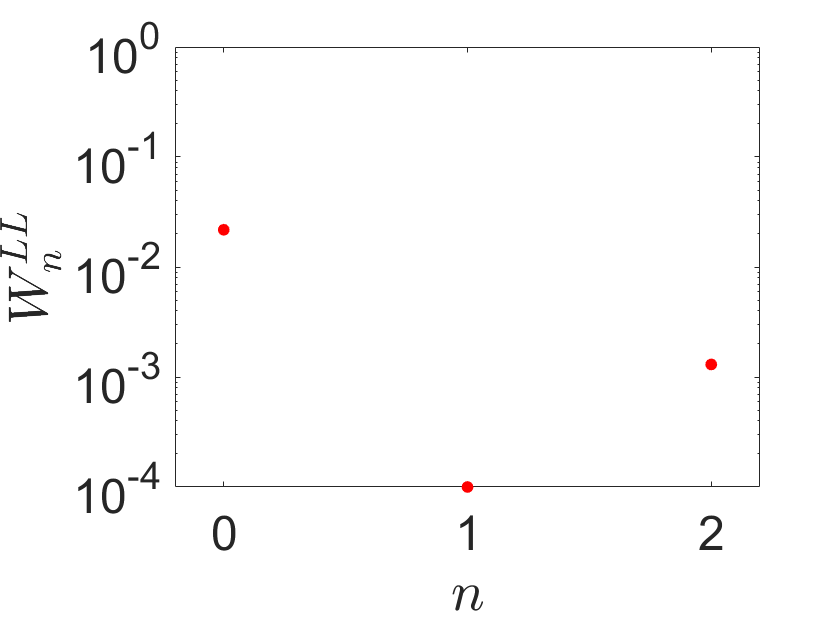}}\\
    \subfigure[0 layer]{\includegraphics[width=0.32 \textwidth]{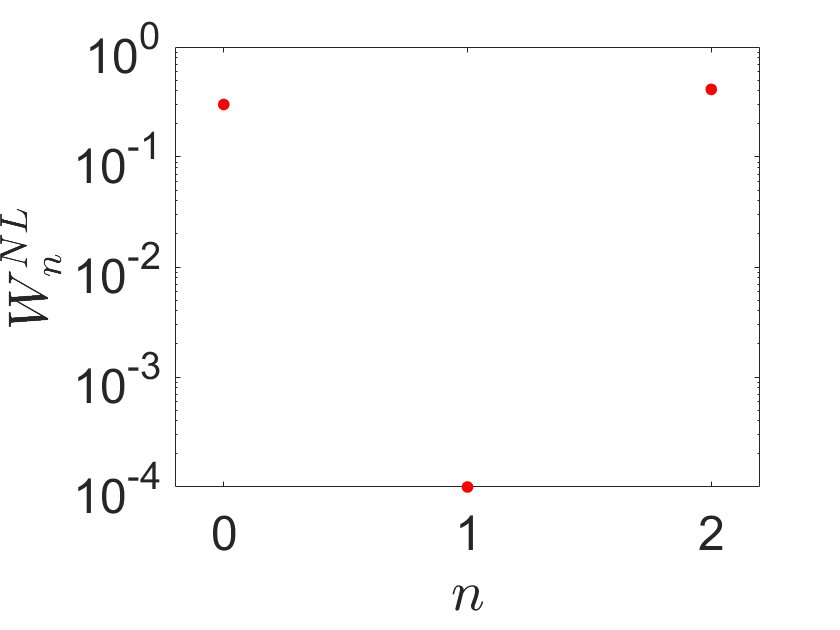}}
    \subfigure[1 layer]{\includegraphics[width=0.32 \textwidth]{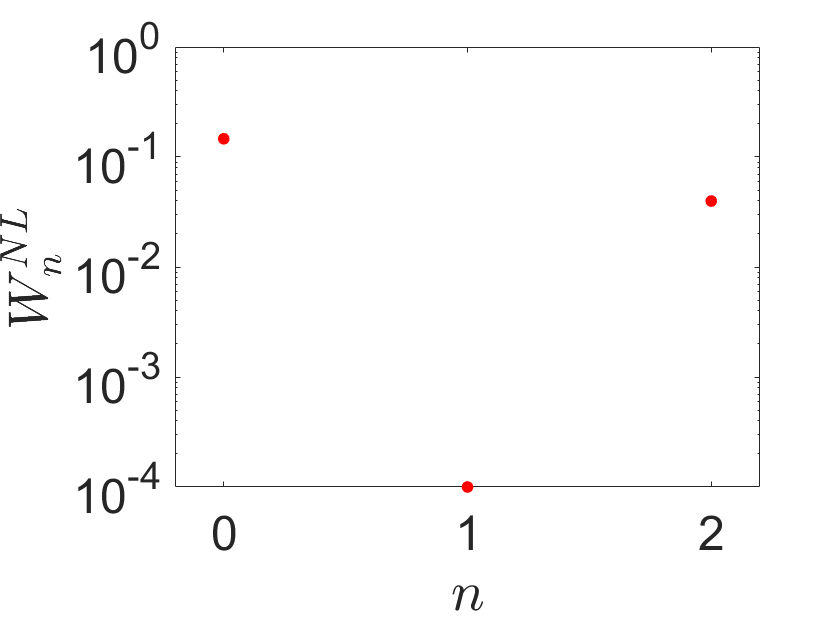}}
    \subfigure[2 layers]{\includegraphics[width=0.32 \textwidth]{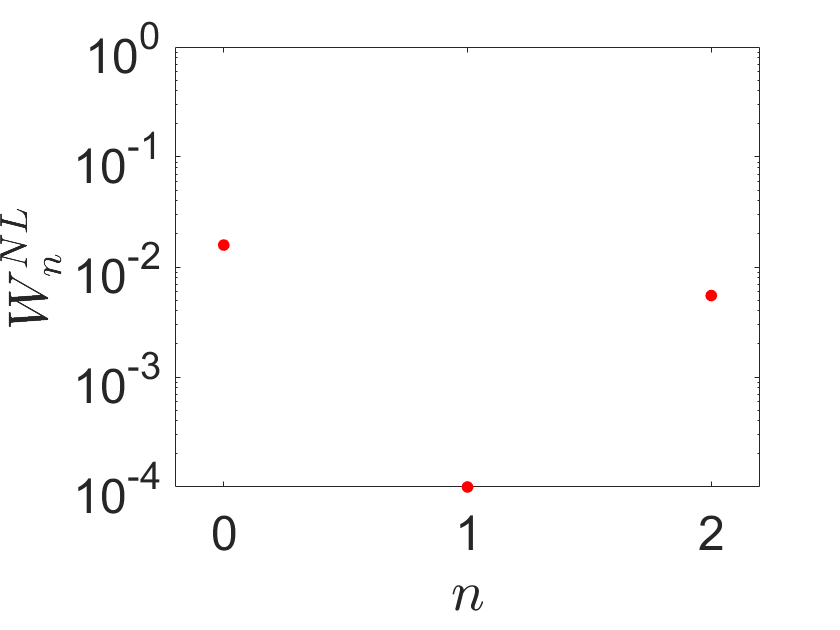}}
    \caption{Scatter plots of the elastic scattering coefficients ($W_n^{LL}$ and $W_n^{NL}$) by using the parameters in Figure \ref{fig:parameter}.} \label{fig:WLL}
\end{figure}

In the first example, we consider an incident wave
\begin{equation*}
  u^{\mathrm{inc}}(\bm x)=\sum_{n=0}^T\sum_{m=-n}^n \bm J_{n,m}^L(\bm x).
\end{equation*}
Figure \ref{fig:parameter} shows the computed material parameters with different layers. The parameters are $\lambda=2.9$, $\mu=0.7$ and $\rho=0.9$ when using one layer, and  $\lambda=(0.1,0.5)$, $\mu=(1.9,0.1)$ and $\rho=(0.1,1.6)$ when using two layers. The elastic scattering coefficients $W_n^{LL}$ and $W_n^{NL}$  are shown in Figure \ref{fig:WLL}. It is clear that the elastic scattering coefficients get smaller as the number of layers increases.

\begin{figure}
    \subfigure[0 layer]{\includegraphics[width=0.32\textwidth]{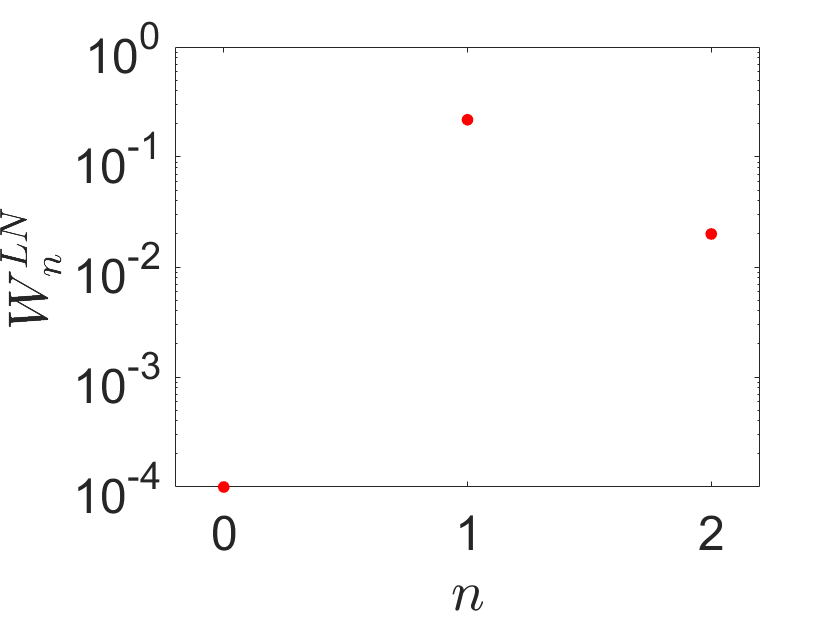}}
    \subfigure[1 layer]{\includegraphics[width=0.32\textwidth]{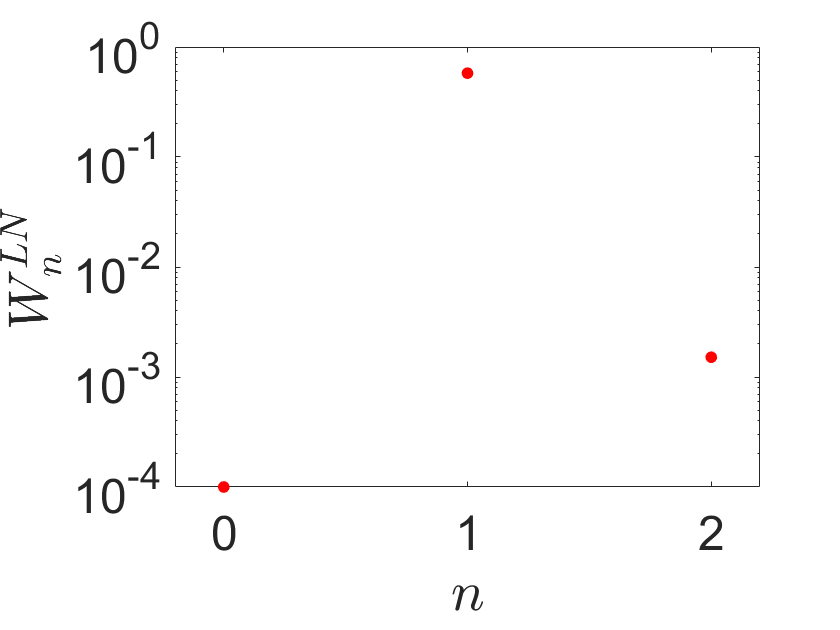}}
    \subfigure[2 layers]{\includegraphics[width=0.32\textwidth]{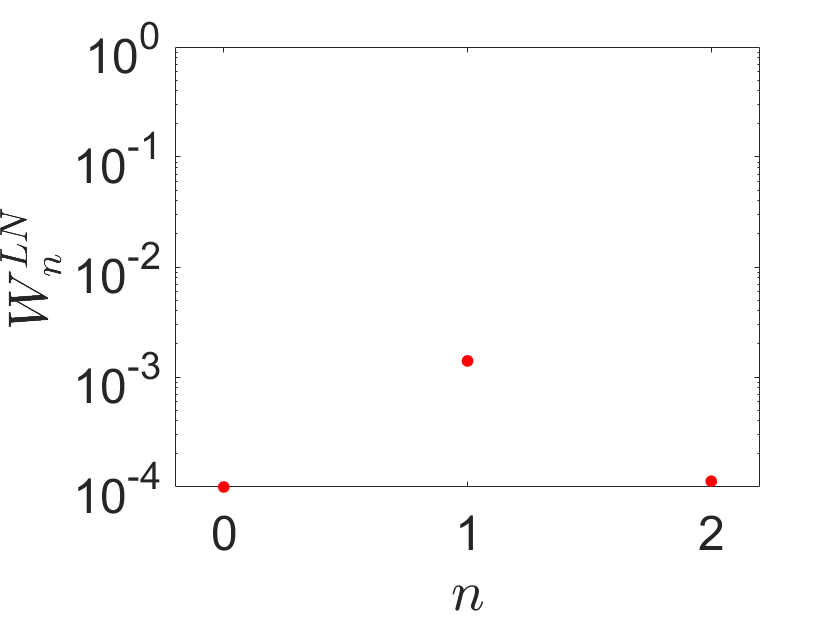}}\\
    \subfigure[0 layer]{\includegraphics[width=0.32 \textwidth]{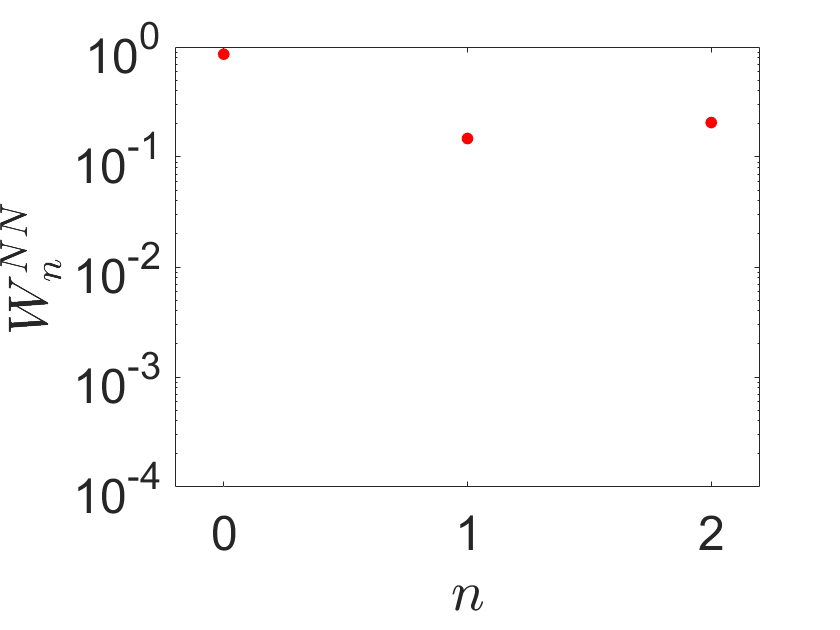}}
    \subfigure[1 layer]{\includegraphics[width=0.32 \textwidth]{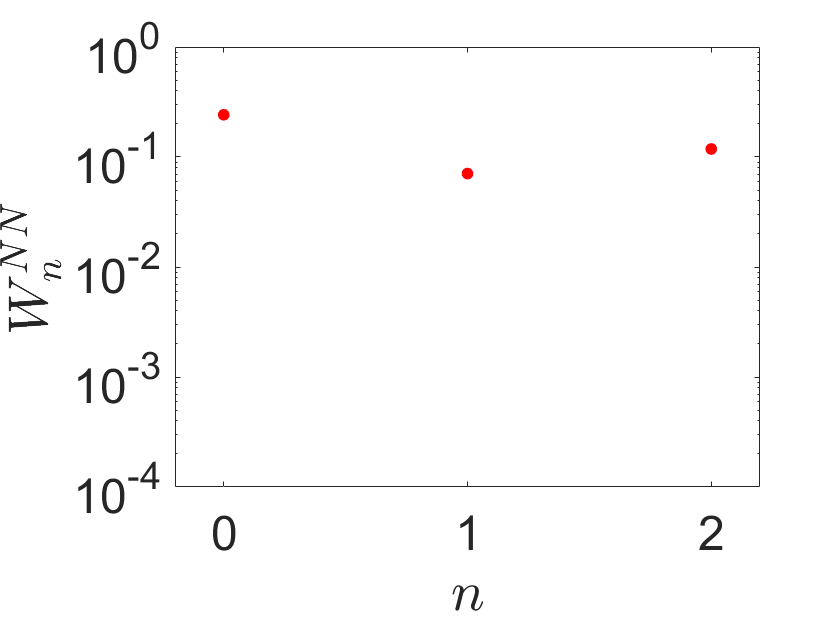}}
    \subfigure[2 layers]{\includegraphics[width=0.32 \textwidth]{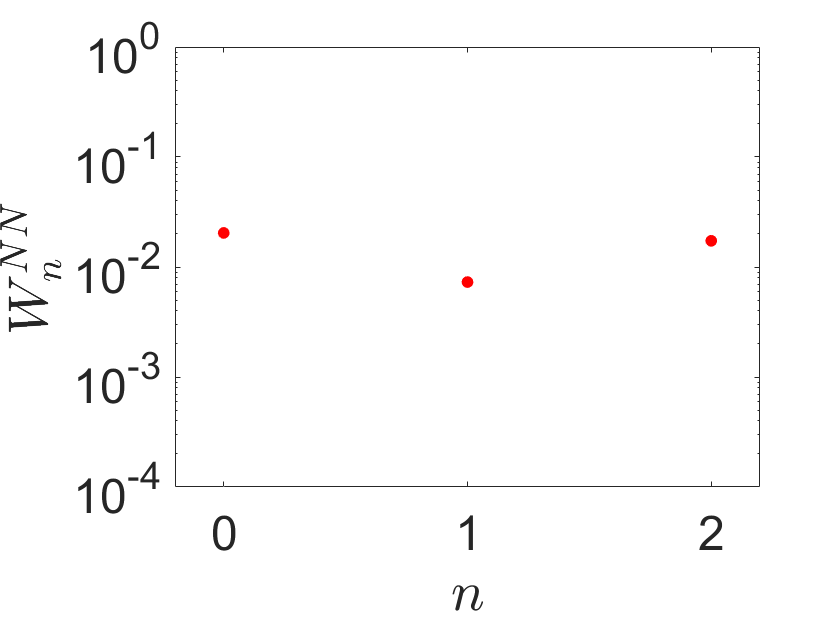}}
    \caption{Scatter plots of the scattering coefficients ($W_n^{LN}$ and $W_n^{NN}$) by using different parameters.} \label{fig:WNN}
\end{figure}

In the second example, we consider an incident wave
\begin{equation*}
  u^{\mathrm{inc}}(\bm x)=\sum_{n=0}^T\sum_{m=-n}^n \bm J_{n,m}^N(\bm x).
\end{equation*}
Through calculation, the material parameters are $\lambda=0.5$, $\mu=0.1$ and $\rho=2.7$ when using one layer, and  $\lambda=(1.9,1.5)$, $\mu=(1.5,1.1)$ and $\rho=(1.1,0.6)$ when using two layers. Figure \ref{fig:WNN} presents the computed elastic scattering coefficients with different layers.

In the final example, we consider an incident wave
\begin{equation*}
  u^{\mathrm{inc}}(\bm x)=\sum_{n=0}^T\sum_{m=-n}^n \bm J_{n,m}^M(\bm x).
\end{equation*}
The computed material parameters are $\lambda=0.1$, $\mu=2.5$ and $\rho=1.2$ when using one layer, and  $\lambda=(0.1, 0.1)$, $\mu=(1.9, 0.1)$ and $\rho=(1.6, 0.1)$ when using two layers. Figure \ref{fig:WMM} presents the computed elastic scattering coefficient $W_n^{MM}$  with different layers.

\begin{figure}
    \subfigure[0 layer]{\includegraphics[width=0.32\textwidth]{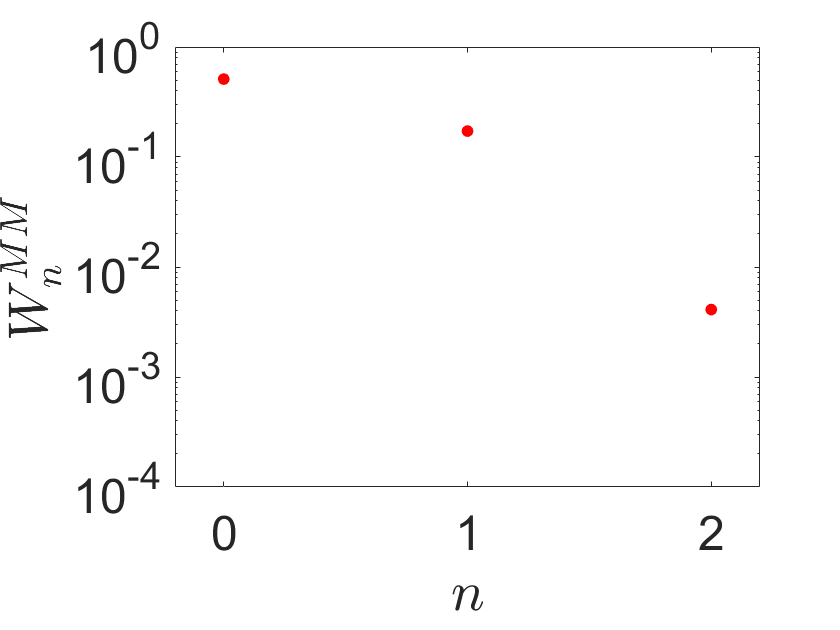}}
    \subfigure[1 layer]{\includegraphics[width=0.32\textwidth]{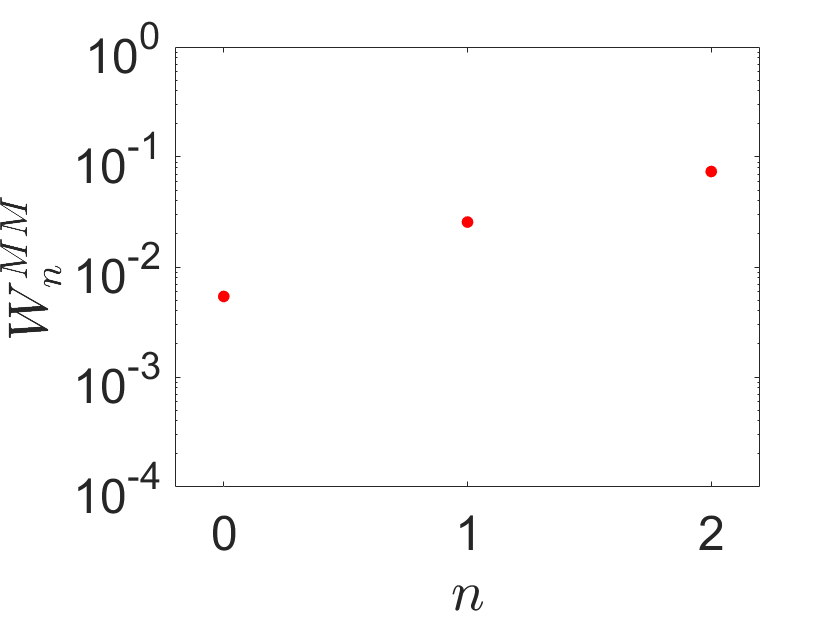}}
    \subfigure[2 layers]{\includegraphics[width=0.32\textwidth]{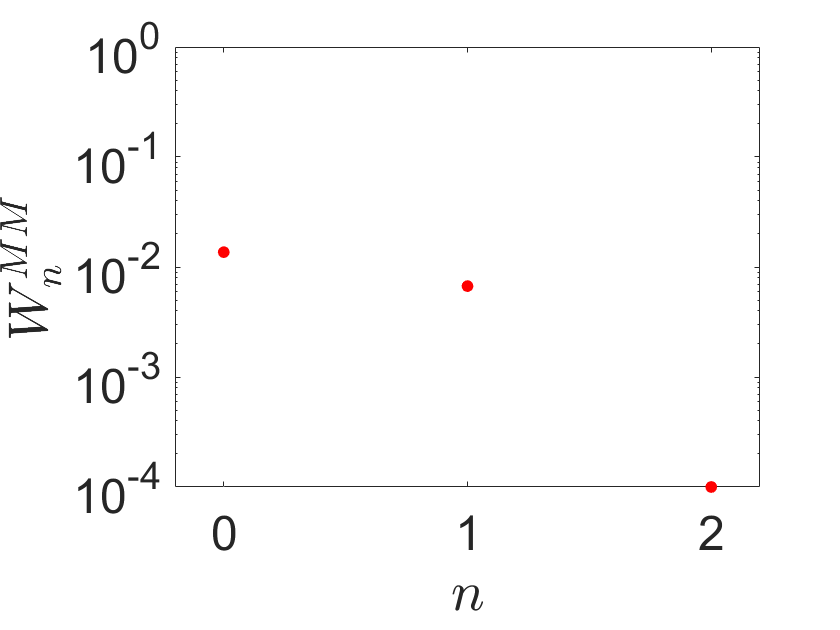}}
    \caption{Scatter plots of the elastic scattering coefficient $W_n^{MM}$ by using different parameters.} \label{fig:WMM}
\end{figure}

\section{Nearly ESC-Vanishing Structures and Enhancement of the Near-Cloaking}\label{sect:Cloaking}

In this section, we  are interested in a nearly ESC-vanishing structure of order $K\in\NN$ at low frequencies in order to construct an effective near-elastic cloaking device. Precisely, we would like to design  $(\lambda,\mu,\rho)$ such that
\begin{align*}
W^{\jmath,\jmath'}_{n}[\lambda, \mu, \rho,\epsilon\omega]=o(\epsilon^{2K+2}), \qquad  \forall\,\jmath,\jmath'=L,M,N,
\end{align*}
for all $n\leq K$ and $\epsilon\leq \epsilon_0$ for some $\epsilon_0\in\RR_+$.
Towards this end, we will first perform a low-frequency asymptotic analysis of the ESC. Then, we will explain the idea of enhancing the capabilities of near-cloaking devices using the concept of nearly ESC-vanishing structures  at low-frequencies.

\subsection{Low-Frequency Behavior of ESC}

Let us  first derive an asymptotic expansion of the entries of $\bR_{n}^{L,N} [\lambda,\mu,\rho,\tau]$ and $\bR_{n}^M [\lambda,\mu,\rho,\tau]$ where $\tau=\epsilon\omega$.
Towards this end, let us recall the series expansions of the spherical Bessel functions of the first and  second kinds (see, e.g., \cite[Eqs. 9.39-9.40]{Monk}),
\begin{align*}
&j_n(\tau)=\sum_{k=0}^{\infty}\frac{{(-1)}^k\tau^{n+2k}}{2^k k!1\cdot3\cdots(2n+2k+1)},
\\
&y_n(\tau)=-\frac{{(2n)}!}{2^nn!}\sum_{k=0}^{\infty}\frac{{(-1)}^k\tau^{2k-n-1}}{2^k k!(-2n+1)(-2n+3)\cdots(-2n+2k+1)}.
\end{align*}
If we use the notation
\begin{equation*}
k!!=\begin{cases}
k\cdot(k-2)\cdots 3\cdot 1  & \text{if}\ k>0\ \text{is odd},
\\
k\cdot(k-2)\cdots 4\cdot 2  & \text{if}\ k>0\ \text{is even},
\\
1 & \text{if}\ k=-1,0,
\end{cases}
\end{equation*}
then
\begin{align}
\label{spherical bessel j}
&j_n(\tau)=\frac{\tau^n}{(2n+1)!!}(1+o(\tau)),&\text{for}\ \tau\ll 1,&
\\
\label{spherical bessel y}
&y_n(\tau)=((2n-1)!!)\tau^{-n-1}(1+o(\tau)), &\text{for}\ \tau\gg 1.&
\end{align}

Based on the behavior of the spherical Bessel functions \eqref{spherical bessel j}-\eqref{spherical bessel y}, the following result holds.  Refer to Appendix \ref{app:lemLF-R} for a sketch of the proof.
\begin{lemma}\label{lemLF-R}
Let  $K\in\NN$ be fixed, $n\in\NN$ be such that $n \leq K$, and  $p,q=1,\cdots, 4$. Then, there exist functions $f^M_{n,k}$, $g^M_{n,k}$,  $f^{L,N}_{n,k, p,q}$, and $g^{L,N}_{n,k, p,q}$ independent of $\tau$ such that, as $\tau=\epsilon\omega\to 0$,
\begin{align}
\begin{cases}
\ds{(R^M_{n})}_{1,1}[\tau]=\tau^n\left(\sum_{k=0}^{K-n}f^M_{n,k}(\lambda,\mu,\rho)\tau^{2k}+o(\tau^{2K-2n})\right),
\\
\ds{(R^M_{n})}_{1,2}[\tau]=\tau^{-n-1}\left(\sum_{k=0}^{K-n}g^M_{n,k}(\lambda,\mu,\rho)\tau^{2k}+o(\tau^{2K-2n}) \right),
\\
\ds(R_{n}^{L,N})_{p,q}[\tau]=\tau^{n-1}\left( \sum_{k=0}^{K-n}f^{L,N}_{n,k,p,q}(\lambda,\mu,\rho)\tau^{2k}+o(\tau^{2K-2n})\right),  & q=1,2,
\\
\ds(R_{n}^{L,N})_{p,q}[\tau]=\tau^{-n-2}\left( \sum_{k=0}^{K-n}g^{L,N}_{n,k,p,q}(\lambda,\mu,\rho)\tau^{2k}+o(\tau^{2K-2n})\right), & q=3,4.
\end{cases}\label{R^M_11}
\end{align}
\end{lemma}

Let us assume that for any $(\lambda,\mu,\rho)$,
\begin{align*}
g^M_{n,0}(\lambda,\mu,\rho)\neq 0
\quad\text{and}\quad
g^{L,N}_{n,0,p,q}(\lambda,\mu,\rho)\neq 0,
\quad \forall p,q=3,4,
\end{align*}
where the functions $g^M_{n,0}$ and $g^{L,N}_{n,0,k,l}$ are given in Lemma \ref{lemLF-R}. Then, the following result can be easily verified using simple algebra.
\begin{theorem}\label{prop}
Let  $K\in\NN$ be fixed and $n\in\NN$ be such that $n \leq K$. Then, there exist functions $W_{n,k}^{\jmath, \jmath'}[\lambda,\mu,\rho]$ and  $W_{n,k}^{M,M}[\lambda,\mu,\rho]$  independent of $\tau$ for $0\leq k\leq K-n$ such that, as $\tau=\epsilon\omega\to 0$,
\begin{equation*}
\begin{aligned}
&\ds W_n^{\jmath,\jmath'}[\lambda,\mu,\rho,\tau]=\tau^{2n}\sum_{k=0}^{K-n}W_{n,k}^{\jmath, \jmath'}[\lambda,\mu,\rho]\tau^{2p}+o(\tau^{2K}),
\\
&\ds W_n^{M,M}[\lambda,\mu,\rho,\tau]=\tau^{2n}\sum_{p=0}^{K-n}W_{n,k}^{M,M}[\lambda,\mu,\rho]\tau^{2k}+o(\tau^{2K}), \qquad  \jmath,\jmath'=L,N.
\end{aligned}
\end{equation*}
\end{theorem}

Notice that, thanks to Lemma \ref{estimate on scattering coefficient}, we have 	
\begin{align*}
\Big| W_{n}^{\jmath,\jmath'} [\lambda, \mu, \rho,\epsilon\omega]\Big|\leq \frac{C^{2n-2}}{n^{2n-2} }\epsilon^{2n-2},\quad \forall \jmath,\jmath'=L,N\quad\text{or}\quad\jmath=\jmath'=M,
\end{align*}
for all $n\in \mathbb{N}$.  Hence, if we can find a structure $(\lambda,\mu,\rho)$ at low frequencies such that the coefficients
\begin{align*}
W_{n,k}^{L,L}=W_{n,k}^{N,L}= W_{n,k}^{N,N}= W_{n,k}^{L,N}=W_{n,k}^{M,M}=0,
\end{align*}
for all $0\leq n\leq K$ and $0\leq k \leq (K-n)$, then
\begin{equation*}
\begin{aligned}
&W_{n}^{\jmath,\jmath'}[\lambda,\mu,\rho,\epsilon\omega]=o(\epsilon^{2K-2}), \quad \jmath,\jmath'=L,N,
\\
&W_{n}^{M,M}[\lambda,\mu,\rho,\epsilon\omega]=o(\epsilon^{2K-2}),
\end{aligned}
\end{equation*}
for all $\epsilon\leq \epsilon_0$ for some $\epsilon_0\in\RR_+$ thanks to Theorem \ref{prop}.
In that case, the far-field signature of the core is also of order $o(\epsilon^{2K-2})$ for all $\epsilon\leq \epsilon_0$. For example, consider an incident wave $\bq e^{\iota\K_S\bx\cdot \bd}$ given by \eqref{incident s-wave}. Then, the scattering amplitudes, defined in  \eqref{far-field pattern}, satisfy
\begin{align*}
\bu^{\infty}_P[\lambda,\mu,\rho,\epsilon\omega]=o(\epsilon^{2K-2})
\quad\text{and}\quad
\bu^\infty_S [\lambda,\mu,\rho,\epsilon\omega]=o(\epsilon^{2K-2}),
\end{align*}
for all $\epsilon\leq \epsilon_0$ thanks to Proposition \ref{PropFarF}. Hence, the application of the nearly ESC-vanishing structure can effectively reduce the \emph{visibility} of the scattering signature of an object.

\subsection{Enhancement of Near-Cloaking in Elasticity}

In this section, we combine the transformation-elastodynamic and the ESC-vanishing structures to enhance the performance of near-elastic cloaking procedures. Let us first briefly review the transformation-elastodynamics and recall the following Lemma  from \cite{hu2015nearly}.

\begin{lemma}\label{lem}
	Let $\widetilde{\bx}=(\widetilde{x}_1,\widetilde{x}_2,\widetilde{x}_3)=\mathcal{F}:\RR^3\to \RR^3$ be an orientation-preserving bi-Lipschitz continuous diffeomorphism such that $\mathcal{F}=\mathcal{I}$ for $|\bx|$ large enough, where $\mathcal{I}$ is the identity map. Then,  $\bu\in H^1(\RR^3)$ is a solution to
\begin{equation*}
\begin{cases}
\nabla\cdot(\fC:\nabla\bu)+\omega^2\rho\bu=0, \quad  \text{in }\, \RR^3,
\\
\bu-\bu^{\rm inc}\quad \text{satisfies Kupradze's radiation condition},
\end{cases}
\end{equation*}
if and only if $\widetilde{\bu}(\widetilde{\by})=\bu\circ\mathcal{F}^{-1}(\widetilde{\by})\in H^1(\RR^3)$ is a solution to
\begin{equation*}
\begin{cases}
\widetilde{\nabla}\cdot(\widetilde{\fC}:\widetilde{\nabla}\widetilde{\bu})+\omega^2\widetilde{\rho}\widetilde{\bu}=0, \quad \text{in }\,\RR^3,
\\
\widetilde{\bu}-\widetilde{\bu}^{\rm inc} \quad \text{ satisfies Kupradze's radiation condition},
\end{cases}
\end{equation*}
where  $\widetilde{\nabla}:=\nabla_{\widetilde{\bx}}$ and $\widetilde{\bu}^{\rm inc}:=\bu^{\rm inc}(\mathcal{F}^{-1})(\widetilde{\by})$. Here, the transformed stiffness tensor and the volume density are given by $\widetilde{\fC}=\left(\widetilde{C}_{ijkl}(\widetilde{x})\right)_{i,j,k,l=1}^3=\mathcal{F}_\ast\fC$ with $\bM={\left({\partial\widetilde{x}_i}/{\partial {x}_j}\right)}^3_{i,j=1}$ and
\begin{align*}
&\widetilde{C}_{ijkl}(\widetilde{x})
=\frac{1}{\det(\bM)}\left.\left( \sum_{p,q=1}^{3} C_{ipkq}\frac{\partial\widetilde{x}_j}{\partial {x}_p}\frac{\partial\widetilde{x}_l}{\partial x_q}\right)\right|_{\bx=\mathcal{F}^{-1}(\widetilde{\by})},
\\
&\widetilde{\rho}=\mathcal{F}_\ast\rho=\left.\left(\frac{\rho}{\det(\bM)}\right)\right|_{\bx=\mathcal{F}^{-1}(\widetilde{\by})}.
\end{align*}
Hence,
\begin{equation*}
\bu^\infty_\alpha[\fC,\rho,\omega]=\widetilde{\bu}^\infty_\alpha [\mathcal{F}_\ast\fC,\mathcal{F}_\ast\rho,\omega], \qquad\alpha=P,S.
\end{equation*}
\end{lemma}

We would like to point out that in Lemma~\ref{lem}, the transformed elastic tensor $\widetilde{C}$ does not possess the minor symmetry required for an elastic tensor, namely $\widetilde{C}_{ijkl}$ lacks the $(i, j)$-symmetry. This poses a significant challenge for the practical value of the transformation-elastodynamics-based cloaks. In fact, it was pointed by Milton et al. \cite{milton2006cloaking} that the invariance of the Lam\'e system can be achieved only if one relaxes the minor symmetry of the transformed elastic tensor. In order to overcome this practical challenge, Norris and Shuvalov \cite{norris2011elastic} and Parnell \cite{parnell2012nonlinear} have explored the elastic cloaking by using Cosserat material or by employing non-linear pre-stress in a neo-Hookean elastomeric material. Designs of transformation-elastodynamics-based Cosserat elastic cloaks without the minor symmetry have been numerically studied in \cite{brun2009achieving,diatta2014controlling}. We refer Remark 2.1 in \cite{hu2015nearly} for more relevant discussions on this point.

In order to apply transformation-elastodynamics, we need to find the scattering amplitudes corresponding to the material parameters beforehand. Towards this end, we introduce the scaling function
\begin{equation*}
\psi_{{1}/{\epsilon}}:\RR^3\to\RR^3,\quad \bx\mapsto{\bx}/{\epsilon},
\end{equation*}
for transforming the small parameter $\epsilon$. Our aim is to obtain the following relation between two scattering amplitudes corresponding to different scaled material parameters and frequency. Consider $\bu$ as a solution to
\begin{equation*}
\begin{cases}
\nabla\cdot(\fC:\nabla \bu(\bx))+\omega^2\rho\bu(x)=\mathbf{0},\quad\text{for}\ \bx\in \RR^3\setminus\overline{B_\epsilon(\mathbf{0})},
\\
\bT[\bu](\bx)=0,\quad\text{on}\ \partial B_\epsilon(\mathbf{0}),
\\
(\bu-\bu^{\rm inc})(\bx)\quad \text{satisfies the Kupradze's radiation condition},
\end{cases}
\end{equation*}
where the incident wave is $\bu^{\rm inc}=\bq e^{\iota\K_S{\bx\cdot \bd}}$ given in \eqref{incident s-wave}. Here, $B_\epsilon(\mathbf{0})$ denotes the ball centered at origin and of radius $\epsilon$. Let $\by={\bx}/{\epsilon}$ and define
\begin{align*}
&\widetilde{\bu}(\by):=\left(\bu\circ\psi_{{1}/{\epsilon}}^{-1}\right)(\by)=(\bu\circ\psi_\epsilon)(\by),
\\
&\widetilde{\bu}^{\rm inc}(\by):=\left(\bu^{\rm inc}\circ\psi_{{1}/{\epsilon}}^{-1}\right)(\by)=(\bu^{\rm inc}\circ\psi_\epsilon)(\by),
\\
&\bT[\widetilde{\bu}](\by):=\lambda\big(\widetilde{\nabla}\cdot\widetilde{\bu}(\by)\big)\widetilde{\bnu}+\mu\big(\widetilde{\nabla}\widetilde{\bu}(\by)+\widetilde{\nabla}\widetilde{\bu}^\top(\by)\big)\widetilde{\bnu}.
\end{align*}
Then, the scattering amplitudes are given by
\begin{align*}
&(\bu-\bu^{\rm inc})(\bx)\thicksim
\frac{e^{\iota\K_P|\bx|}}{\K_P|\bx|}\bu^\infty_P[\fC\circ\psi_{{1}/{\epsilon}},\rho\circ\psi_{{1}/{\epsilon}},\epsilon\omega](\hat{\bx})
\nonumber
\\
&\phantom{(\bu-\bu^{\rm inc})(\bx)\thicksim}
+\frac{e^{\iota\K_S|\bx|}}{\K_S|\bx|}\bu^\infty_S[\fC\circ\psi_{{1}/{\epsilon}},\rho\circ\psi_{{1}/{\epsilon}},\epsilon\omega](\hbx),\quad\text{as }\,  |\bx|\to\infty,
\\
&(\widetilde{\bu}-\widetilde{\bu}^{\rm inc})(\by)\thicksim \frac{e^{\iota\K_P|\by|}}{\K_P|\by|}\bu^\infty_P[\fC,\rho,\epsilon\omega](\hby)+\frac{e^{\iota\K_S|\by|}}{\K_S|\by|}\bu^\infty_S[\fC,\rho,\epsilon\omega](\hby),\,\,\text{as }\, |\by| \to\infty.
\end{align*}
Since  ${(\bu-\bu^{\rm inc})(\bx)}$ and ${(\widetilde{\bu}-\widetilde{\bu}^{\rm inc})(\by)}$ should coincide, we have
\begin{equation*}
\bu^\infty_\alpha[\fC\circ\psi_{{1}/{\epsilon}}\bx,\rho\circ\psi_{{1}/{\epsilon}}(\bx),\omega]
=\bu^\infty_\alpha[\fC,\rho,\epsilon\omega], \qquad\alpha=P,S.
\end{equation*}

Suppose that $(\lambda,\mu,\rho)$ is an ESC-vanishing structure of order $K$ at low frequencies $\epsilon\omega$ for all $\epsilon<\epsilon_0$ for some $\epsilon_0$. Then, we have
\begin{equation}
\label{order of transformed amplitude}
\begin{aligned}
& \bu^\infty_P[\fC\circ\psi_{{1}/{\epsilon}}{(\bx)},\rho\circ\psi_{{1}/{\epsilon}}(\bx),\omega]=o(\epsilon^{2K-2})
=\bu^\infty_S[\fC\circ\psi_{{1}/{\epsilon}}\bx,\rho\circ\psi_{{1}/{\epsilon}}\bx,\omega],
\end{aligned}
\end{equation}
and the diffeomorphism $\mathcal{F}_\epsilon$ is defined by
\begin{equation*}
\mathcal{F}_\epsilon(\bx):=
\begin{cases}
\bx, &\text{for }\, |\bx|\geq 2,
\\
\ds\left(\frac{3-4\epsilon}{2(1-\epsilon)}+\frac{|\bx|}{4(1-\epsilon)}\right){\frac{\bx}{|\bx|}}, &\text{for }\, 2\epsilon\leq |\bx|\leq 2,
\\
\ds\left(\frac{1}{2}+\frac{|\bx|}{2\epsilon}\right)\frac{\bx}{|\bx|}, &\text{for }\, \epsilon\leq |\bx|\leq 2\epsilon,
\\
\ds\frac{\bx}{\epsilon}, &\text{for }\, |\bx|\leq \epsilon.
\end{cases}
\end{equation*}
Therefore, estimate \eqref{order of transformed amplitude},  Proposition \ref{prop}, and Lemma \ref{lem} yield the following key result of this paper.
\begin{theorem}
If $(\lambda,\mu,\rho)$ is a nearly ESC-vanishing structure of order $K$ at low frequencies then there exists $\epsilon_0\in \RR_+$ such that, for all $\epsilon\leq \epsilon_0$,
\begin{align*}
\bu^\infty \left[(\mathcal{F}_\epsilon)_\ast\left(\fC\circ\psi_{{1}/{\epsilon}}(\bx)\right),(\mathcal{F}_\epsilon)_\ast\left(\rho\circ\psi_{{1}/{\epsilon}}(\bx)\right),\omega\right]=o(\epsilon^{2K-2}).
\end{align*}
\end{theorem}

\section{Conclusion}\label{sect:Conc}

 A new method for the design of a near-cloaking structure has been presented at a fixed frequency that enhanced the invisibility effect based on the scheme of vanishing scattering coefficients to elastic scattering problems in three-dimensions.  We have achieved a cloaking effect for an arbitrary elastic object inside the cloaked region with approximately zero scattering amplitudes. Such a cloaking device is obtained by using the transformation-elastodynamics approach of a multi-layered inclusion with a traction-free boundary condition. The cloaking effect for the Lam\'e system is significantly enhanced by the proposed near cloaking structures.

\section*{Acknowledgement}

The research of H Liu was supported by the startup fund from City University of Hong Kong and the Hong Kong RGC General Research Fund ( projects 12302017, 12301218 and 12302919).  The  research of X Wang was supported by the Hong Kong Scholars Program grant  XJ2019005 and the NSFC grant (11971133 and 12001140).

\appendix
\section{Proof of Theorem \ref{symmetry}}\label{app:sym}

Let us first introduce some notation. For all functions $\bv,\bw\in H^{3/2}(D)^3$ and $a,b \in \RR$, we define the quadratic form
\begin{align*}
\langle\bv,\bw\rangle^{a,b}_{D}:=\int_D \left[a(\nabla\cdot\bv)(\nabla\cdot\bw)+\frac{b}{2}(\nabla\bv+\nabla\bv^\top):(\nabla\bw+\nabla\bw^\top)\right]d\bx.
\end{align*}
Here, by slight abuse of notation, we use `$:$' for denoting the matrix contraction, i.e., $\bA:\bB:=\sum_{i,j}a_{ij}b_{ij}$ for all matrices $\bA=(a_{ij})$ and $\bB=(b_{ij})$ having the same dimensions. From the definition, we know that
\begin{equation}\label{green identity}
	\int_{\partial D}\bv\cdot\bT_{a,b}[\bw] d\sigma(\bx)=\int_D\bv\cdot\mathcal{L}_{a,b}[\bw]d\bx+\langle\bv,\bw\rangle^{a,b}_{D},
\end{equation}
where $\bT_{a,b}$ is the surface traction defined in terms of parameters $a,b$.
	If $\bw$ solves the Lam\'e system $\mathcal{L}_{a,b}[\bw]+c\omega^2\bw=\mathbf{0}$, for some $c\in\RR_+$, then
\begin{align*}
\int_{\partial D} \bv\cdot\bT_{a,b}[\bw]d\sigma(\bx)=-c\omega^2\int_{\partial D}\bv\cdot\bw d\bx+\langle\bv,\bw\rangle^{a,b}_{D},
\end{align*}
and hence from (\ref{green identity}), we have
\begin{align}
\int_{\partial D}\bv\cdot\bT_{a,b}[\bw]d\sigma(\bx)=&\int_{\partial D}\bT_{a,b}[\bv]\cdot\bw d\sigma(\bx)-c\omega^2\int_{\partial D}\bv\cdot\bw d\bx
-\int_D\mathcal{L}_{a,b}[\bv]\cdot\bw d\bx.\label{green plus w}
\end{align}
In addition, if $\bv$ is a solution of $\mathcal{L}_{a,b}[\bv]+c\omega^2\bv=\mathbf{0}$ then
\begin{equation}\label{commute}
\int_{\partial D}\bv\cdot\bT_{a,b}[\bw]d\sigma(\bx)=\int_{\partial D}\bT_{a,b}[\bv]\cdot\bw d\sigma(\bx).
\end{equation}
Henceforth, we use the notation
\begin{equation*}
\begin{cases}
\ds\eta_L:=\frac{\mu_0}{\mu_1-\mu_0},
\\
\ds\widetilde{\eta}_L:=\frac{\mu_1}{\mu_1-\mu_0},
\\
\ds\eta_M=\eta_N:=\frac{3\lambda_0+2\mu_0}{3(\lambda_1-\lambda_0)+2(\mu_1-\mu_0)},
\\
\ds\widetilde{\eta}_M=\widetilde{\eta}_N:=\frac{3\lambda_1+2\mu_1}{3(\lambda_1-\lambda_0)+2(\mu_1-\mu_0)}.
\end{cases}
\end{equation*}

Let $(\bphi_{nm}^{\jmath},\bpsi_{nm}^{\jmath})$ and $(\bphi_{kl}^{\jmath'},\bpsi_{kl}^{\jmath'})$ be, respectively, the solutions to
\begin{align}
\begin{aligned}
&\ds\widetilde{\mathcal{S}}^\omega_D[\bphi_{nm}^{\jmath}]-\mathcal{S}^\omega_D[\bpsi_{nm}^{\jmath}]=\bJ^{\jmath}_{nm}\big|_{\partial D},
\\
&\ds \widetilde{\bT}\big[\widetilde{\mathcal{S}}^\omega_D[\bphi_{nm}^{\jmath}]\big]\Big|_-
-\bT\left[\mathcal{S}^\omega_D[\bpsi_{nm}^{\jmath}]\right]\Big|_+=\ds\bT\left[\bJ^{\jmath}_{nm}\right]\Big|_{\partial D},
\end{aligned}\label{alpha(sym)}
\end{align}
and
\begin{align}
\begin{aligned}
&\ds\widetilde{\mathcal{S}}^\omega_D[\bphi_{kl}^{\jmath'}]-\mathcal{S}^\omega_D[\bpsi_{kl}^{\jmath'}]=\bJ^{\jmath'}_{kl}\big|_{\partial D},
\\
&\ds\widetilde{\bT}\left[\widetilde{\mathcal{S}}^\omega_D[\bphi_{kl}^{\jmath'}]\right]\Big|_-
-\bT\left[\mathcal{S}^\omega_D[\bpsi_{kl}^{\jmath'}]\right]\Big|_+=\ds\bT\left[\bJ^{\jmath'}_{kl}\right] \Big|_{\partial D}.
\end{aligned}
\label{beta(sym)}
\end{align}
Thanks to jump conditions \eqref{Sjumps}, the scattering coefficients $W_{(n,m),(k,l)}^{\jmath,\jmath'}$ can be expressed as
\begin{align}
W_{(n,m),(k,l)}^{\jmath,\jmath'}=&\ds\int_{\partial D} \overline{\bJ_{nm}^{\jmath}}\cdot \bphi_{kl}^{\jmath'}d\sigma(\bx)
\nonumber
\\
=&\ds \int_{\partial D} \overline{\bJ_{nm}^{\jmath}}\cdot\left[{\bT\left[\mathcal{S}^\omega_D[\bphi^{\jmath'}_{kl}]\right]}\Big|_+ -{\bT\left[\mathcal{S}^\omega_D[\bphi^{\jmath'}_{kl}]\right]}\Big|_-\right] d\sigma(\bx).
\label{intermid2}
\end{align}
Using \eqref{beta(sym)} in \eqref{intermid2} and then invoking \eqref{green plus w}-\eqref{commute}, one gets
\begin{align*}
W_{(n,m),(k,l)}^{\jmath,\jmath'}
=&-\int_{\partial D}\overline{\bJ_{nm}^{\jmath}}\cdot\bT\left[\bJ_{kl}^{\jmath'}\right]d\sigma(\bx)
\\
&+\int_{\partial D}\overline{\bJ_{nm}^{\jmath}}\cdot\left[\widetilde{\bT}\left[\widetilde{\mathcal{S}}_D^{\omega}[\bpsi^{\jmath'}_{kl}]\right]\Big|_{-} -\bT\left[\mathcal{S}^\omega_D[\bphi^{\jmath'}_{kl}]\right]\Big|_{+}\right]d\sigma(\bx),
\end{align*}
or equivalently
\begin{align*}
W_{(n,m),(k,l)}^{\jmath,\jmath'}
=&-\int_{\partial D}\overline{\bJ_{nm}^{\jmath}}\cdot\bT\left[\bJ_{kl}^{\jmath'}\right]d\sigma(\bx)
\\
&+\int_{\partial D}\left[\widetilde{\bT}\left[\overline{\bJ_{nm}^{\jmath}}\right]\cdot\widetilde{\mathcal{S}}_D^{\omega}[\bpsi^{\jmath'}_{kl}]-\bT\left[\overline{\bJ_{nm}^{\jmath}}\right]\cdot \mathcal{S}_D^{\omega}[\bpsi^{\jmath'}_{kl}]\right]d\sigma(\bx)
\\
&-\rho_1\omega^2\int_{D}\overline{\bJ_{nm}^{\jmath}}\cdot\widetilde{\mathcal{S}}_D^{\omega}[\bpsi^{\jmath'}_{kl}]d\bx-\int_D\mathcal{L}_{\lambda_1,\mu_1}[\overline{\bJ_{nm}^{\jmath}}]\cdot\widetilde{\mathcal{S}}_D^{\omega}[\bpsi^{\jmath'}_{kl}]d\bx.
\end{align*}
Invoking  the first equation of \eqref{beta(sym)}, this leads to
\begin{align*}
W_{(n,m),(k,l)}^{\jmath,\jmath'}=
&
-\int_{\partial D}\overline{\bJ_{nm}^{\jmath}}\cdot\bT\left[\bJ_{kl}^{\jmath'}\right]d\sigma(\bx)+\int_{\partial D}\widetilde{\bT}\left[\overline{\bJ_{nm}^{\jmath}}\right]\cdot\widetilde{\mathcal{S}}_D^{\omega}[\bpsi^{\jmath'}_{kl}]d\sigma(\bx)
\\
&-\int_{\partial D}\bT\left[\overline{\bJ_{nm}^{\jmath}}\right]\cdot\widetilde{\mathcal{S}}_D^{\omega}[\bpsi^{\jmath'}_{kl}]d\sigma(\bx)
+\int_{\partial D}\bT\left[\overline{\bJ_{nm}^{\jmath}}\right]\cdot\bJ_{kl}^{\jmath'}d\sigma(\bx)
\\
&-\rho_1\omega^2\int_{D}\overline{\bJ_{nm}^{\jmath}}\cdot\widetilde{\mathcal{S}}_D^{\omega}[\bpsi^{\jmath'}_{kl}]d\bx-\int_D\mathcal{L}_{\lambda_1,\mu_1}[\overline{\bJ_{nm}^{\jmath}}]\cdot\widetilde{\mathcal{S}}_D^{\omega}[\bpsi^{\jmath'}_{kl}]d\bx.
\end{align*}
Note that the first term and the fourth term on the right hand side of the above equation cancel out each other thanks to (\ref{commute}). Therefore,
\begin{align}
W_{(n,m),(k,l)}^{\jmath,\jmath'}
=&
\int_{\partial D}\left[\widetilde{\bT}\left[\overline{\bJ_{nm}^{\jmath}}\right]-{\bT}\left[\overline{\bJ_{nm}^{\jmath}}\right]\right]\cdot\widetilde{\mathcal{S}}_D^{\omega}[\bpsi^{\jmath'}_{kl}]d\sigma(\bx)
\nonumber
\\
&-\rho_1\omega^2\int_{D}\overline{\bJ_{nm}^{\jmath}}\cdot\widetilde{\mathcal{S}}_D^{\omega}[\bpsi^{\jmath'}_{kl}]d\bx
-\int_D\mathcal{L}_{\lambda_1,\mu_1}[\overline{\bJ_{nm}^{\jmath}}]\cdot\widetilde{\mathcal{S}}_D^{\omega}[\bpsi^{\jmath'}_{kl}]d\bx.
\label{using commute W}
\end{align}

Remark that, by the definition of the surface traction operator and the spherical wave functions $\bJ^\jmath_{nm}$, it can be derived after fairly easy manipulations that
\begin{equation}
\label{surface traction operation}
\widetilde{\bT}\left[\overline{\bJ_{nm}^{\jmath}}\right]-{\bT}\left[\overline{\bJ_{nm}^{\jmath}}\right]=\frac{1}{\eta_\jmath}{\bT}\left[\overline{\bJ_{nm}^{\jmath}}\right]
\qquad\text{and}\qquad
\widetilde{\bT}\left[\overline{\bJ_{nm}^{\jmath}}\right]-{\bT}\left[\overline{\bJ_{nm}^{\jmath}}\right]=\frac{1}{\widetilde{\eta}_\jmath}\widetilde{\bT}\left[\overline{\bJ_{nm}^{\jmath}}\right].
\end{equation}
Therefore, using Eq. \eqref{green plus w} and the second relation from \eqref{surface traction operation}, one gets
\begin{align*}
\widetilde{\eta}_\jmath  W_{(n,m),(k,l)}^{\jmath,\jmath'}
=&\int_{\partial D}\widetilde{\bT}\left[\bJ_{nm}^{\jmath}\right]\cdot\widetilde{\mathcal{S}}_D^{\omega}[\bpsi^{\jmath'}_{kl}]d\sigma(\bx)
-\widetilde{\eta}_\jmath\rho_1\omega^2\int_{D}\overline{\bJ_{nm}^{\jmath}}\cdot\widetilde{\mathcal{S}}_D^{\omega}[\bpsi^{\jmath'}_{kl}]d\bx
\\
&-\widetilde{\eta}_\jmath\int_D\mathcal{L}_{\lambda_1,\mu_1}[\overline{\bJ_{nm}^{\jmath}}]\cdot\widetilde{\mathcal{S}}_D^{\omega}[\bpsi^{\jmath'}_{kl}]d\bx
\\
=&\int_{\partial D}\overline{\bJ_{nm}^{\jmath}}\cdot\widetilde{\bT}\left[\widetilde{\mathcal{S}}_D^{\omega}[\bpsi^{\jmath'}_{kl}]\right]\Big|_- d\sigma(\bx)
+(1-\widetilde{\eta}_\jmath)\rho_1\omega^2\int_{D}\overline{\bJ_{nm}^{\jmath}}\cdot\widetilde{\mathcal{S}}_D^{\omega}[\bpsi^{\jmath'}_{kl}]d\bx
\\
&+(1-\widetilde{\eta}_\jmath)\int_D\mathcal{L}_{\lambda_1,\mu_1}[\overline{\bJ_{nm}^{\jmath}}]\cdot\widetilde{\mathcal{S}}_D^{\omega}[\bpsi^{\jmath'}_{kl}]d\bx.
\end{align*}
This, together with \eqref{alpha(sym)}-\eqref{beta(sym)}, renders
\begin{align}
\widetilde{\eta}_\jmath W_{(n,m),(k,l)}^{\jmath,\jmath'}
=&\int_{\partial D}\overline{\widetilde{\mathcal{S}}^\omega_D[\bpsi^{\jmath}_{nm}]}\cdot\widetilde{\bT}\left[\widetilde{\mathcal{S}}_D^{\omega}[\bpsi^{\jmath'}_{kl}]\right]\big|_-d\sigma(\bx)
-\int_{\partial D}\overline{\widetilde{\mathcal{S}}^\omega_D[\bphi^{\jmath}_{nm}]}\cdot\bT\left[\widetilde{\mathcal{S}}_D^{\omega}[\bphi^{\jmath'}_{kl}]\right]\big|_+d\sigma(\bx)
\nonumber
\\
&-\int_{\partial D}\overline{\widetilde{\mathcal{S}}^\omega_D[\bphi^{\jmath}_{nm}]}\cdot\bT\left[\bJ_{kl}^{\jmath'}\right]d\sigma(\bx) d\sigma(\bx)
+(1-\widetilde{\eta}_\jmath)\rho_1\omega^2\int_{D}\overline{\bJ_{nm}^{\jmath}}\cdot\widetilde{\mathcal{S}}_D^{\omega}[\bpsi^{\jmath'}_{kl}]d\bx
\nonumber
\\&+(1-\widetilde{\eta}_\jmath)\int_D\mathcal{L}_{\lambda_1,\mu_1}[\overline{\bJ_{nm}^{\jmath}}]\cdot\widetilde{\mathcal{S}}_D^{\omega}[\bpsi^{\jmath'}_{kl}]d\bx.
\label{eta alpha beta}
\end{align}

Similarly, on substituting the first relation from \eqref{surface traction operation} into \eqref{using commute W} and using the first equation of \eqref{beta(sym)}, we arrive at
\begin{align}
{\eta}_\jmath	W_{(n,m),(k,l)}^{\jmath,\jmath'}
=&\int_{\partial D}\bT\left[\overline{\bJ_{nm}^{\jmath}}\right]\cdot\widetilde{\mathcal{S}}_D^{\omega}[\bpsi^{\jmath'}_{kl}]d\sigma(\bx)
-\eta_\jmath\rho_1\omega^2\int_D \overline{\bJ_{nm}^{\jmath}}\cdot\widetilde{\mathcal{S}}_D^{\omega}[\bpsi^{\jmath'}_{kl}]d\bx
\nonumber
\\
&-\eta_\jmath \int_D\mathcal{L}_{\lambda_1,\mu_1}[\overline{\bJ_{nm}^{\jmath}}]\cdot\widetilde{\mathcal{S}}_D^{\omega}[\bpsi^{\jmath'}_{kl}]d\bx
\nonumber
\\
=&\int_{\partial D}\bT\left[\overline{\bJ_{nm}^{\jmath}}\right]\cdot \mathcal{S}_D^{\omega}[\bpsi^{\jmath'}_{kl}]d\sigma(\bx)
+\int_{\partial D}\bT\left[\overline{\bJ_{nm}^{\jmath}}\right]\cdot \overline{\bJ_{kl}^{\jmath'}}d\sigma(\bx)
\nonumber
\\
&-\eta_\jmath\rho_1\omega^2\int_D \overline{\bJ_{nm}^{\jmath}}\cdot\widetilde{\mathcal{S}}_D^{\omega}[\bpsi^{\jmath'}_{kl}]d\bx
-\eta_\jmath \int_D\mathcal{L}_{\lambda_1,\mu_1}[\overline{\bJ_{nm}^{\jmath}}]\cdot\widetilde{\mathcal{S}}_D^{\omega}[\bpsi^{\jmath'}_{kl}]d\bx.\label{eta last}
\end{align}

Finally, subtracting \eqref{eta last} from \eqref{eta alpha beta} and noting that $\widetilde{\eta}_\jmath-\eta_\jmath=1$, we get
\begin{align}
W_{(n,m),(k,l)}^{\jmath,\jmath'}
=&\int_{\partial D}\overline{\widetilde{\mathcal{S}}^\omega_D[\bpsi^{\jmath}_{nm}]}\cdot\widetilde{\bT}\left[\widetilde{\mathcal{S}}_D^{\omega}[\psi^{\jmath'}_{kl}]\right]\big|_-d\sigma(\bx)
-\int_{\partial D}\overline{\mathcal{S}^\omega_D[\bphi^{\jmath}_{nm}]}\cdot\bT\left[\mathcal{S}_D^{\omega}[\bphi^{\jmath'}_{kl}]\right]|_+d\sigma(\bx)
\nonumber
\\
& -\int_{\partial D}\overline{\mathcal{S}_D^{\omega}[\bphi^{\jmath}_{nm}]}\cdot\bT\left[\bJ_{kl}^{\jmath'}\right] d\sigma(\bx)
-\int_{\partial D}\bT\left[\overline{\bJ_{nm}^{\jmath}}\right]\cdot \mathcal{S}_D^{\omega}[\bphi^{\jmath'}_{kl}]d\sigma(\bx)
\nonumber
\\
& -\int_{\partial D}\bT\left[\overline{\bJ_{nm}^{\jmath}}\right]\cdot \overline{\bJ_{kl}^{\jmath'}}d\sigma(\bx).
\label{substracting}
\end{align}

By proceeding in the same fashion as above, $W_{(n,m),(k,l)}^{\jmath',\jmath}$ can be expressed as
\begin{align*}
W_{(k,l),(n,m)}^{\jmath',\jmath}
=&\int_{\partial D}\overline{\widetilde{\mathcal{S}}^\omega_D[\bpsi^{\jmath'}_{kl}]}\cdot\widetilde{\bT}\left[\widetilde{\mathcal{S}}_D^{\omega}[\bpsi^{\jmath}_{nm}]\right]\big|_-d\sigma(\bx)
-\int_{\partial D}\overline{\mathcal{S}^\omega_D[\bphi^{\jmath'}_{kl}]}\cdot\bT\left[\mathcal{S}_D^{\omega}[\bphi^{\jmath}_{nm}]\right]\big|_+ d\sigma(\bx)
\\
&-\int_{\partial D}\overline{\mathcal{S}_D^{\omega}[\bphi^{\jmath'}_{kl}]}\cdot\bT\left[\bJ_{nm}^{\jmath}\right] d\sigma(\bx)
-\int_{\partial D}\bT\left[\overline{\bJ_{kl}^{\jmath'}}\right]\cdot \mathcal{S}_D^{\omega}[\bphi^{\jmath}_{nm}]d\sigma(\bx)
\\
&-\int_{\partial D}\bT\left[\overline{\bJ_{kl}^{\jmath'}}\right]\cdot \overline{\bJ_{nm}^{\jmath}}d\sigma(\bx),
\end{align*}
or equivalently,
\begin{align}
W_{(k,l),(n,m)}^{\jmath',\jmath}
=&\int_{\partial D}\widetilde{\bT}\left[\overline{\widetilde{\mathcal{S}}_D^{\omega}[\bpsi^{\jmath}_{kl}]}\right]\big|_-\cdot \widetilde{\mathcal{S}}^\omega_D[\bpsi^{\jmath}_{nm}]d\sigma(\bx)
\nonumber
\\
&-\int_{\partial D}\bT\left[\overline{\mathcal{S}^\omega_D[\bphi^{\jmath'}_{kl}]}\right]\big|_+\cdot \mathcal{S}_D^{\omega}[\bphi^{\jmath}_{nm}]d\sigma(\bx)
-\int_{\partial D}\overline{\mathcal{S}_D^{\omega}[\bphi^{\jmath'}_{kl}]}\cdot\bT\left[\bJ_{nm}^{\jmath}\right]d\sigma(\bx)
\nonumber
\\
&-\int_{\partial D}\bT\left[\overline{\bJ_{kl}^{\jmath'}}\right]\cdot \mathcal{S}_D^{\omega}[\bphi^{\jmath}_{nm}]d\sigma(\bx)
-\int_{\partial D}\overline{\bJ_{kl}^{\jmath}}\cdot\bT\left[\overline{\bJ_{nm}^{\jmath'}}\right] d\sigma(\bx).
\label{last}
\end{align}
The proof is completed by taking complex conjugate of expression \eqref{last} and comparing the result with equation \eqref{substracting}.

\section{Surface Tractions of Multipole Elastic Fields}\label{app:tractions}

In spherical coordinate system, one can use vector spherical harmonics to get the surface tractions \cite[Eq. 13.3.78, p. 1872]{morse1953methods}
\begin{equation*}
\begin{aligned}
\bT[\bH^L_{nm}] =&2\mu_0 \K_P\Bigg\{\left[\left(n^2+n-\frac{\K_S^2r^2}{2} \right) \frac{ h_n^{(1)}(\K_P r)}{{(\K_P r)}^2} -2\frac{ ((h_n^{(1)} (\K_P r))^\prime }{\K_P r}\right] \bA_{nm}
 \\
 &
 +\left[\frac{h_n^{(1)} (\K_P r)}{\K_P r}\right]^\prime \sqrt {n(n+1)}\bB_{nm}\Bigg\},
 \\
\bT[\bH^M_{nm}] =&\mu_0 \K_S\sqrt{n(n+1)}\left[{\left(h^{(1)}_n (\K_Sr)\right)}^\prime-\frac{h^{(1)}_n(\K_Sr)}{\K_Sr}\right]\bC_{nm},
\\
\bT[\bH^N_{nm}] =&2\mu_0 \K_S\sqrt{n(n+1)}\Bigg\{\left[\frac{h^{(1)}_n(\K_Sr)}{\K_Sr}\right]^\prime n(n+1)\bA_{nm}\\&+\left[\left(n^2+n-1-\frac{\K_S^2r^2}{2}\right)\frac{h^{(1)}_n(\K_Sr)}{(\K_Sr)^2}-\frac{(h^{(1)}_n)^\prime(\K_Sr)}{\K_Sr}\right]\bB_{nm}\Bigg\}.
\end{aligned}
\end{equation*}
The surface traction of the interior multipole elastic fields $\bJ^\jmath$ can be obtained by replacing $h_n^{(1)}$ with $j_n$.

\section{Expressions of  Sub-Matrices}\label{app:Ps}

Let  $[\bP_{n}^{L,N}]_\ell$ be defined by
$ \left([\bP_{n}^{L,N}]_\ell(r)\right)_{pq}:= \left(P^{L,N}_{\ell, n}\right)_{p,q}(r)$, for all  $p,q=1,\cdots, 4$. Then, the elements ${(P^{L,N}_{\ell, n})}_{p,q}$ are given by
\begin{align*}
&{(P^{L,N}_{\ell, n})}_{1,1}(r)= (j_n)^\prime(\K_{P,\ell}r),
\qquad
{(P^{L,N}_{\ell, n})}_{1,2}(r) = \frac{n(n+1)}{\K_{S,\ell}r}j_n(\K_{S,\ell}r),
\\
&{(P^{L,N}_{\ell, n})}_{2,1}(r) = \frac{j_n(\K_{P,\ell}r)}{\K_{P,\ell}r},
\qquad\quad
{(P^{L,N}_{\ell, n})}_{2,2}(r)= \frac{\mathcal{J}_n(\K_{S,\ell}r)}{\K_{S,\ell}r},
\end{align*}
\begin{align*}
&{(P^{L,N}_{\ell, n})}_{3,1}(r) = \mu_\ell \K_{P,\ell}\left[ \left(n^2+n-\frac{(\K_{S,\ell }r)^2}{2} \right ) \frac{j_n(\K_{P,\ell}r)}{{(\K_{P,\ell}r)}^2} -\frac{2(j_n)^\prime(\K_{P,\ell}r_\ell)}{\K_{P,\ell}r} \right],
 \\
& {(P^{L,N}_{\ell, n})}_{3,2}(r)= \mu_\ell \K_{S,\ell}\left[\frac{j_n(\K_{S,\ell}r)}{\K_{S,\ell}r} \right ]^\prime \sqrt{n(n+1)},
\end{align*}
\begin{align*}
 &{(P^{L,N}_{\ell, n})}_{4,1}(r)= \mu_\ell \K_{P,\ell}\left[\frac{j_n(\K_{P,\ell}r)}{\K_{P,\ell}r} \right ]^\prime \sqrt{n(n+1)},
\\
&{(P^{L,N}_{\ell, n})}_{4,2}(r)=\mu_\ell \K_{S,\ell}\left[ \left(n^2+n-1-\frac{(\K_{S,\ell}r)^2}{2} \right ) \frac{j_n(\K_{S,\ell}r)}{{(\K_{S,\ell}r)}^2}-\frac{{(j_n)^\prime}(\K_{S,\ell}r)}{\K_{S,\ell}r} \right],
\end{align*}
whereas  ${(P^{L,N}_{\ell, n})}_{p,q}$,  for $p=1,\cdots,4,$ and $q=3,4$, can be defined by replacing the Bessel functions $j_n$ in ${(P^{L,N}_{\ell, n})}_{p,q-2}$ with Hankel functions $h_n^{(1)}$.

Similarly, if we define the sub-matrix $[\bP_{n}^{M}]_\ell$ by
$\left([\bP_{n}^{M}]_\ell(r)\right)_{pq}:=  (P^{M}_{\ell, n})_{p,q}(r)$,  for all $p,q=1,\cdots,4$, then the elements $(P^{M}_{\ell, n})_{p,q}$ are given by
\begin{align*}
&(P^{M}_{\ell, n})_{1,1}(r) = j_n(\K_{S,\ell}r),
\\
&(P^{M}_{\ell, n})_{1,2}(r) =  h^{(1)}_n(\K_{S,\ell}r),
 \\
&(P^{M}_{\ell, n})_{2,1}(r) = \mu_{\ell}\K_{S,\ell}\left[\left(j_n\right)'\left(\K_{S,\ell}r \right)-\frac{j_n\left(\K_{S,\ell}r\right)}{\K_{S,\ell}r}\right],
\\
&(P^{M}_{\ell, n})_{2,2(r)} = \mu_{\ell}\K_{S,\ell}\left[\left(h^{(1)}_n\right)'\left(\K_{S,\ell}r\right)-\frac{h^{(1)}_n\left(\K_{S,\ell}r\right)}{\K_{S,\ell} r}\right].
\end{align*}

Let the matrix $[\bQ_{n}^{L,N}]_\L$ be given by $\left([\bQ_{n}^{L,N}]_\L\right)_{pq}=  \left(Q^{L,N}_{\L, n}\right)_{p,q}$, for all $p,q=1,\cdots,4$. Then,
\begin{align*}
{(Q^{L,N}_{\L, n})}_{1,1} =& \left(n^2+n-\frac{1}{2}\K_{S,{\L}}^2\right)\frac{j_n(\K_{P,{\L}})}{\K_{P,{\L}}}-2{(j_n)}^\prime(\K_{P,{\L}}),
\\
{(Q^{L,N}_{{\L}, n})}_{1,2} =& \K_{S,{\L}}\left[ \frac{j_n(\K_{S,{\L}})}{\K_{S,{\L}}} \right]^\prime \sqrt{n(n+1)},
\\
{(Q^{L,N}_{{\L}, n})}_{2,1} =& \K_{P,{\L}}\left[ \frac{j_n(\K_{P,{\L}})}{\K_{P,{\L}}} \right]^\prime \sqrt{n(n+1)},
\\
(Q^{L,N}_{{\L}, n})_{2,2} =&   \left( n^2+n-1-\frac{1}{2}\K_{S,{\L}}^2\right)\frac{j_n(\K_{S,{\L}})}{\K_{S,{\L}}}-{(j_n)}^\prime(\K_{S,{\L}}),
\\
(Q^{L,N}_{{\L}, n})_{p,q} =& 0, \qquad p=3,4 \quad\text{and}\quad q=1,\cdots, 4,
\end{align*}
and $(Q^{L,N}_{{\L}, n})_{p,q}$, for $p=1,2$ and $q=3,4$, can be defined by replacing $j_n$ in ${(Q^{L,N}_{{\L}, n})}_{p,q-2}$  with  $h_n^{(1)}$.
Finally, the matrix $[\bQ_{n}^{M}]_{\L}$ is given by
\begin{align*}
[\bQ_{n}^{M}]_{\L}
=\begin{bmatrix}
\K_{S,{\L}}{\left(j_n\right)}^\prime(\K_{S,{\L}})-j_n(\K_{S,{\L}})  &  \K_{S,{\L}}\left(h_n^{(1)}\right)^\prime (\K_{S,{\L}})-h_n^{(1)}(\K_{S,{\L}})
\\\\
0&0
\end{bmatrix}.
\end{align*}


\section{Proof of Lemma \ref{lemLF-R}}\label{app:lemLF-R}

We provide a sketch of the proof of Lemma \ref{lemLF-R} here. Towards this end, we first study the low-frequency behavior of sub-matrices $[\bP^{L,N}_{n}]_\L$, $[\bP^{M}_{n}]_\L$, $[\bQ^{L,N}_{n}]_\L$, and $[\bQ^{M}_{n}]_\L$. To facilitate the ensuing analysis, we introduce the shorthand notation
$\phi_n=(2n-1)!!$,  and $t_{\alpha,\ell}=1/c_{\alpha,\ell}$, for $\alpha=P,S$,
so that $\epsilon\K_{\alpha,\ell}=\tau t_{\alpha,\ell}$.

Thanks to the behavior of first and second kind Bessel functions for small inputs as described in \eqref{spherical bessel j} - \eqref{spherical bessel y}, the entries of the matrices $[\bQ^{L,N}_{n}]_\L$ become
\begin{align*}
	{(Q^{L,N}_{\L, n})}_{1,1}&=\frac{n(n-1){(t_{P,\L})}^{n-1}}{\phi_{n+1}}\tau^{n-1}+o(\tau^{n-1}),\\
	{(Q^{L,N}_{\L, n})}_{1,2}&=\frac{\sqrt{n(n+1)}(n-1){( t_{S,\L} )}^{n-1}}{\phi_{n+1}}\tau^{n-1}+o(\tau^{n-1}),
\end{align*}
\begin{align*}
	{(Q^{L,N}_{\L, n})}_{1,3}&=-\frac{\iota \phi(n)(n+1)(n+2)}{{( t_{P,\L} )}^{n+2}}\tau^{-n-2}+o(\tau^{-n-2}),\\
	{(Q^{L,N}_{\L, n})}_{1,4}&=\frac{\iota \phi(n)\sqrt{n(n+1)}(n+2)}{{( t_{S,\L} )}^{n+2}}\tau^{-n-2}+o(\tau^{-n-2}),
\end{align*}
\begin{align*}
	{(Q^{L,N}_{\L, n})}_{2,1}&=\frac{\sqrt{n(n+1)}(n-1){( t_{P,\L} )}^{n-1}}{\phi_{n+1}}\tau^{n-1}+o(\tau^{n-1}),\\
	{(Q^{L,N}_{\L, n})}_{2,2}&=\frac{n^2{( t_{S,\L} )}^{n-1}}{\phi_{n+1}}\tau^{n-1}+o(\tau^{n-1}),
\end{align*}
\begin{align*}
	{(Q^{L,N}_{\L, n})}_{2,3}&=\frac{\iota \phi(n)\sqrt{n(n+1)}(n+2)}{{( t_{P,\L} )}^{n+2}}\tau^{-n-2}+o(\tau^{-n-2}),\\
	{(  Q^{L,N}_{\L, n})}_{2,4}&=-\frac{\iota \phi(n){(n+1)}^2}{{( t_{S,\L} )}^{n+2}}\tau^{-n-2}+o(\tau^{-n-2}).
\end{align*}

Similarly, the entries of $[\bP^{L,N}_{n}]_\ell$ become
\begin{align*}
{(P^{L,N}_{\ell, n})}_{1,1}&=\frac{n{(t_{P,\ell}r_\ell)}^{n-1}}{\phi_{n+1}}\tau^{n-1}+o(\tau^{n-1}),
\\
{(P^{L,N}_{\ell, n})}_{1,2}&=\frac{n(n+1){(t_{S,\ell} r_\ell)}^{n-1}}{\phi_{n+1}}\tau^{n-1}+o(\tau^{n-1}),
\end{align*}
\begin{align*}
{(P^{L,N}_{\ell, n})}_{1,3}&=\frac{\iota \phi(n)(n+1)}{{(t_{P,\ell}r_\ell)}^{n+2}}\tau^{-n-2}+o(\tau^{-n-2}),
\\
{(P^{L,N}_{\ell, n})}_{1,4}&=-\frac{\iota n(n+1)\phi(n)}{{(t_{S,\ell}r_\ell)}^{n+2}}\tau^{-n-2}+o(\tau^{-n-2}),
\end{align*}
\begin{align*}
{(P^{L,N}_{\ell, n})}_{2,1}&=\frac{{(t_{P,\ell}r_\ell)}^{n-1}}{\phi_{n+1}}\tau^{n-1}+o(\tau^{n-1}),
\\
{(P^{L,N}_{\ell, n})}_{2,2}&=\frac{(n+1){(t_{S,\ell}r_\ell)}^{n-1}}{\phi_{n+1}}\tau^{n-1}+o(\tau^{n-1}),
\end{align*}
\begin{align*}
{(P^{L,N}_{\ell, n})}_{2,3}&=-\frac{\iota \phi(n)}{{(t_{P,\ell} r_\ell)}^{n+2}}\tau^{-n-2}+o(\tau^{-n-2}),
\\
{(P^{L,N}_{\ell, n})}_{2,4}&=-\frac{\iota \phi(n)n}{{(t_{S,\ell}r_\ell)}^{n+2}}\tau^{-n-2}+o(\tau^{-n-2}),
\end{align*}
\begin{align*}
{(P^{L,N}_{\ell, n})}_{3,1}&=\frac{\mu_\ell n(n-1){(t_{P,\ell}r_\ell)}^{n-1}}{r_\ell\phi_{n+1}}\tau^{n-1}+o(\tau^{n-1}),
\\
{(P^{L,N}_{\ell, n})}_{3,2}&=\frac{\mu_\ell\sqrt{n(n+1)}(n-1){(t_{S,\ell}r_\ell)}^{n-1}}{r_\ell\phi_{n+1}}\tau^{n-1}+o(\tau^{n-1}),
\end{align*}
\begin{align*}
{(P^{L,N}_{\ell, n})}_{3,3}&=-\frac{\iota \mu_\ell \phi(n)(n^2+3n+2)}{r_\ell{(t_{P,\ell}r_\ell)}^{n+2}}\tau^{-n-2}+o(\tau^{-n-2}),
\\
{(P^{L,N}_{\ell, n})}_{3,4}&=\frac{\iota \mu_\ell \phi(n)\sqrt{n(n+1)}(n+2)}{r_\ell{(t_{S,\ell}r_\ell)}^{n+2}}\tau^{-n-2}+o(\tau^{-n-2}),
\end{align*}
\begin{align*}
{(P^{L,N}_{\ell, n})}_{4,1}&=\frac{\mu_\ell\sqrt{n(n+1)}(n-1){(t_{P,\ell}r_\ell)}^{n-1}}{r_\ell\phi_{n+1}}\tau^{n-1}+o(\tau^{n-1}),
\\
{(P^{L,N}_{\ell, n})}_{4,2}&=\frac{\mu_\ell n^2{(t_{S,\ell}r_\ell)}^{n-1}}{r_\ell\phi_{n+1}}\tau^{n-1}+o(\tau^{n-1}),
\end{align*}
\begin{align*}
{(P^{L,N}_{\ell, n})}_{4,3}&=\frac{\iota \mu_\ell \phi(n)\sqrt{n(n+1)}(n+2)}{r_\ell{(t_{P,\ell}r_\ell)}^{n+2}}\tau^{-n-2}+o(\tau^{-n-2}),
\\
{(P^{L,N}_{\ell, n})}_{4,4}&=-\frac{\iota \mu_\ell \phi(n)(n+1)^2}{r_\ell{(t_{S,\ell}r_\ell)}^{n+2}}\tau^{-n-2}+o(\tau^{-n-2}).
\end{align*}

Note that
\begin{align*}
\det\big([\bP^{L,N}_{n}]_\ell\big)
=&
\frac{{(\mu_\ell \phi(n))}^2}{\phi_{n+1}^2t_{P,\ell}^3t_{S,\ell}^3r^8_\ell}\tau^{-6}\Bigg[(-4n^6-14n^5-10n^4+8n^3+11n^2+3n)
\\&
+\sqrt{n(n+1)}(2n^5+9n^4+8n^3-4n^2-5n-1)\Bigg]+o(\tau^{-6}).
\end{align*}
Therefore, the entries of matrix $[\bP^{L,N}_{n}]_\ell^{-1}[\lambda,\mu,\rho,\tau]=\left(p_{i,j}\right)_{i,j=1}^4$ are given by
\begin{align*}
   p_{1,q}&=(-1)^{q}\frac{\xi_{1q}(n)}{(r_\ell t_{P,\ell})^{n-1}}\tau^{-n+1}+o(\tau^{-n+1}), \quad q=1,2,
\\
 p_{1,q}&=(-1)^{q}\frac{\xi_{1q}(n)r_\ell}{\mu_\ell(r_\ell t_{P,\ell})^{n-1}}\tau^{-n+1}+o(\tau^{-n+1}), \quad q=3,4,
\end{align*}
\begin{align*}
 p_{2,q}&=(-1)^{q+1}\frac{\xi_{2q}(n)}{(r_\ell t_{S,\ell})^{n-1}}\tau^{-n+1}+o(\tau^{-n+1}), \quad q=1,2,
\\
 p_{2,q}&=\frac{\xi_{2q}(n)r_\ell}{\mu_\ell(r_\ell t_{S,\ell})^{n-1}}\tau^{-n+1}+o(\tau^{-n+1}), \quad q=3,4,
\end{align*}
\begin{align*}
   p_{3,q}&=-\iota\xi_{3q}(n){(r_\ell t_{P,\ell})}^{n+2}\tau^{n+2}+o(\tau^{n+2}), \quad q=1,2,
\\
   p_{3,q}&=\iota\xi_{3q}(n)\frac{r_\ell(r_\ell t_{P,\ell})^{n+2}}{\mu_\ell }\tau^{n+2}+o(\tau^{n+2}), \quad q=3,4,
\end{align*}
\begin{align*}
   p_{4,q}&=(-1)^q \iota\xi_{4q}(n){(r_\ell t_{S,\ell})}^{n+2}\tau^{n+2}+o(\tau^{n+2}), \quad q=1,2,
\\
  p_{4,q}&=(-1)^q\iota\xi_{4q}(n)\frac{r_\ell (r_\ell t_{S,\ell})^{n+2}}{\mu_\ell}\tau^{n+2}+o(\tau^{n+2}), \quad q=3,4,
\end{align*}	
for some function $\left(\xi_{ij}(n)\right)_{i,j=1}^4$ of $n$, independent of $\tau$.

Having obtained the low-frequency behavior of $[\bP^{L,N}_{n}]_\ell^{-1}[\lambda,\mu,\rho,\tau]$ and $[\bP^{L,N}_{n}]_{\ell-1}[\lambda,\mu,\rho,\tau]$,  one can obtain the behavior of the matrix $[\bS^{L,N}_{n}]_\ell[\lambda,\mu,\rho,\tau]:=[\bP^{L,N}_{n}]_\ell^{-1}[\bP^{L,N}_{n}]_{\ell-1}$, for $\ell=1,2,\cdots, \L$. Indeed, if $[\bS^{L,N}_{n}]_\ell$ is defined as
\begin{equation*}
{[\bS^{L,N}_{n}]}_\ell[\lambda,\mu,\rho,\tau]:=\left(s_{p,q}\right)_{p,q=1}^4.
\end{equation*}
Then the entries of matrix ${[\bS^{L,N}_{n}]}_\ell$ are given by
\begin{align*}
s_{1,1}=&{\left(\frac{t_{P,\ell-1}}{t_{P,\ell}}\right)}^{n-1}\frac{1}{\phi_{n+1}}
\Bigg(-\xi_{11}n+\xi_{12}-\xi_{13}\left(\frac{\mu_{\ell-1}}{\mu_\ell}\right)n(n-1)
\\
&+\xi_{14}\left(\frac{\mu_{\ell-1}}{\mu_\ell}\right)\sqrt{n(n+1)}(n-1)\Bigg)(1+o(1)),
\\
s_{1,2}=&{\left(\frac{t_{S,\ell-1}}{t_{P,\ell}}\right)}^{n-1}\frac{1}{\phi_{n+1}}
\Bigg(-\xi_{11}n(n+1)+\xi_{12}(n+1)-\xi_{13}\left(\frac{\mu_{\ell-1}}{\mu_\ell}\right)\sqrt{n(n+1)}(n-1)
\\
&+\xi_{14}\left(\frac{\mu_{\ell-1}}{\mu_\ell}\right)n^2\Bigg)(1+o(1)),
\end{align*}
\begin{align*}
s_{1,3}=&\frac{\iota \phi(n)}{{(r_\ell t_{P,\ell})}^{n-1}{(r_\ell t_{P,\ell-1})}^{n+2}}
\Bigg(-\xi_{11}(n+1)+\xi_{13}\left(\frac{\mu_{\ell-1}}{\mu_\ell}\right)(n+1)(n+2)
\\
&-\xi_{12}+\xi_{14}\left(\frac{\mu_{\ell-1}}{\mu_\ell}\right)\sqrt{n(n+1)}(n+2)\Bigg)\tau^{-2n-1}(1+o(1)),
\\
s_{1,4}=&\frac{\iota \phi(n)}{{(r_\ell t_{P,\ell})}^{n-1}{(r_\ell t_{S,\ell-1})}^{n+2}}\Bigg(\xi_{11}n(n+1)
-\xi_{13}\left(\frac{\mu_{\ell-1}}{\mu_\ell}\right)\sqrt{n(n+1)}(n+2)
\\
&-\xi_{12}n-\xi_{14}\left(\frac{\mu_{\ell-1}}{\mu_\ell}\right)(n+1)^2\Bigg)\tau^{-2n-1}(1+o(1)),
\end{align*}
\begin{align*}
s_{2,1}=&{\left(\frac{t_{P,\ell-1}}{t_{S,\ell}}\right)}^{n-1}\frac{1}{\phi_{n+1}}
\Bigg(\xi_{21}n-\xi_{22}-\xi_{23}\left(\frac{\mu_{\ell-1}}{\mu_\ell}\right)n(n-1)
\\
&+\xi_{24}\left(\frac{\mu_{\ell-1}}{\mu_\ell}\right)\sqrt{n(n+1)}(n-1)\Bigg)(1+o(1)),
\\
s_{2,2}=&{\left(\frac{t_{S,\ell-1}}{t_{S,\ell}}\right)}^{n-1}\frac{1}{\phi_{n+1}}\Bigg(\xi_{21}n(n+1)-\xi_{22}(n+1)+\xi_{23}\left(\frac{\mu_{\ell-1}}{\mu_\ell}\right)\sqrt{n(n+1)}(n-1)
\\
&+\xi_{24}\left(\frac{\mu_{\ell-1}}{\mu_\ell}\right)n^2\Bigg)(1+o(1)),
\end{align*}
\begin{align*}
 s_{2,3}=&\frac{\iota \phi(n)}{{(r_\ell t_{S,\ell})}^{n-1}{(r_\ell t_{P,\ell-1})}^{n+2}}\Bigg(\xi_{21}(n+1)-\xi_{23}\left(\frac{\mu_{\ell-1}}{\mu_\ell}\right)(n+1)(n+2)
 \\
&+\xi_{22} +\xi_{24}\left(\frac{\mu_{\ell-1}}{\mu_\ell}\right)\sqrt{n(n+1)}(n+2)
 \Bigg)\tau^{-2n-1}(1+o(1)),
\\
s_{2,4}=&\frac{\iota \phi(n)}{(r_\ell t_{S,\ell})^{n-1}{(r_\ell t_{S,\ell-1})}^{n+2}}\Bigg(\xi_{22}n+\xi_{23}\left(\frac{\mu_{\ell-1}}{\mu_\ell}\right)\sqrt{n(n+1)}(n+2)
\\
&-\xi_{21}n(n+1)-\xi_{14}\left(\frac{\mu_{\ell-1}}{\mu_\ell}\right)(n+1)^2\Bigg)\tau^{-2n-1}(1+o(1)),
\end{align*}	
\begin{align*}	
s_{3,1}=&\iota\frac{{(r_\ell t_{P,\ell})}^{n+2}{(r_\ell t_{P,\ell-1})}^{n-1}}{\phi_{n+1}}\Bigg(-\xi_{31}n-\xi_{32}+\xi_{33}\left(\frac{\mu_{\ell-1}}{\mu_\ell}\right)n(n-1)
\\
&+\xi_{34}\left(\frac{\mu_{\ell-1}}{\mu_\ell}\right)\sqrt{n(n+1)}(n-1)\Bigg)\tau^{2n+1}(1+o(1)),
\\
s_{3,2}=&\iota\frac{{(r_\ell t_{P,\ell})}^{n+2}{(r_\ell t_{S,\ell-1})}^{n-1}}{\phi_{n+1}}\Bigg(-\xi_{31}n(n+1)-\xi_{32}(n+1)
\\
&+\xi_{33}\left(\frac{\mu_{\ell-1}}{\mu_\ell}\right)\sqrt{n(n+1)}(n-1)+\xi_{34}\left(\frac{\mu_{\ell-1}}{\mu_\ell}\right)n^2\Bigg)\tau^{2n+1}(1+o(1)),
\end{align*}
\begin{align*}
s_{3,3}=&-\phi(n){\left(\frac{t_{P,\ell}}{t_{P,\ell-1}}\right)}^{n+2}\Bigg(-\xi_{31}(n+1)-\xi_{33}\left(\frac{\mu_{\ell-1}}{\mu_\ell}\right)(n+1)(n+2)
\\
&+\xi_{32}+\xi_{34}\left(\frac{\mu_{\ell-1}}{\mu_\ell}\right)\sqrt{n(n+1)}(n+2)\Bigg)(1+o(1)),
\\
s_{3,4}=&-\phi(n){\left(\frac{t_{P,\ell}}{t_{S,\ell-1}}\right)}^{n+2}\Bigg(\xi_{31}n(n+1)+\xi_{32}n+\xi_{33}\left(\frac{\mu_{\ell-1}}{\mu_\ell}\right)\sqrt{n(n+1)}(n+2)
\\
&-\xi_{34}\left(\frac{\mu_{\ell-1}}{\mu_\ell}\right)(n+1)^2\Bigg)(1+o(1)),
\end{align*}
\begin{align*}
s_{4,1}=&\iota\frac{{(r_\ell t_{S,\ell})}^{n+2}{(r_\ell t_{P,\ell-1})}^{n-1}}{\phi_{n+1}}\Bigg(-\xi_{41}n+\xi_{42}+\xi_{43}\left(\frac{\mu_{\ell-1}}{\mu_\ell}\right)n(n-1)
\\
&+\xi_{44}\left(\frac{\mu_{\ell-1}}{\mu_\ell}\right)\sqrt{n(n+1)}(n-1)\Bigg)\tau^{2n+1}(1+o(1)),
\\
s_{4,2}=&\iota\frac{{(r_\ell t_{S,\ell})}^{n+2}{(r_\ell t_{S,\ell-1})}^{n-1}}{\phi_{n+1}}\Bigg(-\xi_{41}n(n+1)+\xi_{42}(n+1)
\\
&+\xi_{43}\left(\frac{\mu_{\ell-1}}{\mu_\ell}\right)\sqrt{n(n+1)}(n-1)+\xi_{44}\left(\frac{\mu_{\ell-1}}{\mu_\ell}\right)n^2\Bigg)\tau^{2n+1}(1+o(1)),
\end{align*}
\begin{align*}
s_{4,3}=&-\phi(n){\left(\frac{t_{S,\ell}}{t_{P,\ell-1}}\right)}^{n+2}\Bigg(-\xi_{41}(n+1)-\xi_{42}-\xi_{43}\left(\frac{\mu_{\ell-1}}{\mu_\ell}\right)(n+1)(n+2)
\\
&+\xi_{44}\left(\frac{\mu_{\ell-1}}{\mu_\ell}\right)\sqrt{n(n+1)}(n+2)\Bigg)(1+o(1)),
\\
s_{4,4}=&-\phi(n){\left(\frac{t_{S,\ell}}{t_{S,\ell-1}}\right)}^{n+2}\Bigg(\xi_{41}n(n+1)-\xi_{42}n+\xi_{43}\left(\frac{\mu_{\ell-1}}{\mu_\ell}\right)\sqrt{n(n+1)}(n+2)
\\
&-\xi_{44}\left(\frac{\mu_{\ell-1}}{\mu_\ell}\right)(n+1)^2\Bigg)(1+o(1)).
\end{align*}		

Finally, using the asymptotic forms of the element of $[\bP^{L,N}_{n}]_\L$, $[\bP^{M}_{n}]_\L$, $[\bQ^{L,N}_{n}]_\L$, $[\bQ^{M}_{n}]_\L$, and $[\bS^{L,N}_{n}]_\ell$ in \eqref{H^M}, on can find the low-frequency behavior of the matrices  $\bR_{n}^{M}$, $\bR_{n}^{L,N}$, expressed in the required form  \eqref{R^M_11}. This completes the proof.

 \bibliographystyle{plain}

\end{document}